\documentclass[11pt,titlepage]{article}
\pdfoutput=1 
\usepackage{natbib}
\usepackage{amsmath, amsthm, amssymb}
\usepackage{graphicx}
\usepackage[margin=3.3cm]{geometry}
\usepackage{multirow}
\long\def\comment#1{} 

\newtheorem{thm}{Theorem}[section]
\newtheorem{cor}[thm]{Corollary}
\newtheorem{lem}[thm]{Lemma}
\newtheorem{dfn}{Definition}

\long\def\comment#1{}    
\newcommand{\centerfig}[2]{\centerline{\dofig{#1}{#2}}}
\newcommand{\dofig}[2]{\resizebox{#1}{!}{\includegraphics{#2}}}

\begin{document}

\newpage

\begin{center}
{\LARGE\bf Ensemble Patch Transformation: A New Tool for Signal Decomposition 
\medskip
}
\vskip 7mm

{\large\sc Donghoh Kim$^{1}$, Guebin Choi$^{2}$ and Hee-Seok Oh$^{2}$}\\
{\large $^{1}$Sejong University, Seoul, Korea, $^{2}$Seoul National University, Seoul, Korea}
\end{center}
\vskip 5mm

\noindent 
{\bf Abstract}: This paper considers the problem of signal decomposition and data visualization. For this purpose, we introduce a new multiscale transform, termed `ensemble patch transformation' that enhances identification of local characteristics embedded in a signal and provides multiscale visualization according to different levels; hence, it is useful for data analysis and signal decomposition. In literature, there are data-adaptive decomposition methods such as empirical mode decomposition (EMD) by Huang et al. (1998). Along the same line of EMD, we propose a new decomposition algorithm that extracts meaningful components from a signal that belongs to a large  class of signals, compared to the previous methods. Some theoretical properties of the proposed algorithm are investigated. To evaluate the proposed method, we analyze several synthetic examples and a real-world signal.
%
%
%
%


\pagenumbering{arabic}

\pagenumbering{arabic}

%
\section{Introduction}
In this paper, we propose a new multiscale method for data analysis and signal decomposition, termed `ensemble patch transformation', which adopts a multiscale concept of scale-space theory in computer vision of \cite{Lindeberg1994}. The proposed ensemble patch transformation consists of two key concepts. The first one is `patch process' that is defined as a data-dependent patch at certain time point $t$ for a given sequence of data.  The patch process is designed for identifying dependent structures of data according to various sizes of patches. The second concept is `ensemble' that is obtained by shifting the time point $t$ for the patch, which is suitable for representing temporal variation of data efficiently by enhancement of temporal resolution of them. Moreover, it is feasible that the proposed ensemble patches provide various statistics; hence, these can be easily adapted for various purposes of data analysis.

We focus on the problem of signal decomposition and extraction using the proposed ensemble patch transformation. A successful recognition of the local frequency patterns of a signal is a crucial step for signal decomposition. Empirical mode decomposition (EMD) by Huang et al. (1998) identifies such local patterns through local extrema. In the case that the local extrema reflect the time-varying amplitude and frequency, EMD decomposes a signal effectively according to its frequencies. However, when the high-frequency pattern is not distinct in a signal, EMD fails to identify a superimposed component; thus, it produces artificial components during the decomposition process. To clarify this problem of EMD and provide a motivation of the proposed method, we consider a synthetic signal that consists of two components $X_t = \cos(100 \pi t) + 4 \cos(60 \pi t), \mbox{ }t \in[0, 1]$. Figure~\ref{signal3} shows signal $X_t$ and its two components. The middle panel of Figure~\ref{signal3decom} illustrates the decomposed results produced by EMD, where the dotted lines represent true components and the solid lines are extracted components. As one can see, EMD fails to decompose the two components of the signal where the frequency ratio of two components is relatively small. In other words, when the local pattern of the high-frequency component is not distinct, EMD does not work properly to decompose a signal; hence, it fails to extract the sinusoid components effectively. We remark that Rilling and Flandrin (2008) discussed the ranges of frequency and amplitude ratios when EMD performs for decomposition of signals. On the other hand, the left panel of Figure~\ref{signal3decom} presents the decomposition results by the proposed method in Section~\ref{sec.proposal}, which identify the true components efficiently. The right panel shows the decomposition results by ensemble EMD (EEMD) of Wu and Huang (2009), which cannot extract the true ones properly.   
\begin{figure}[!t]
\centering
\includegraphics[width=0.85\linewidth]{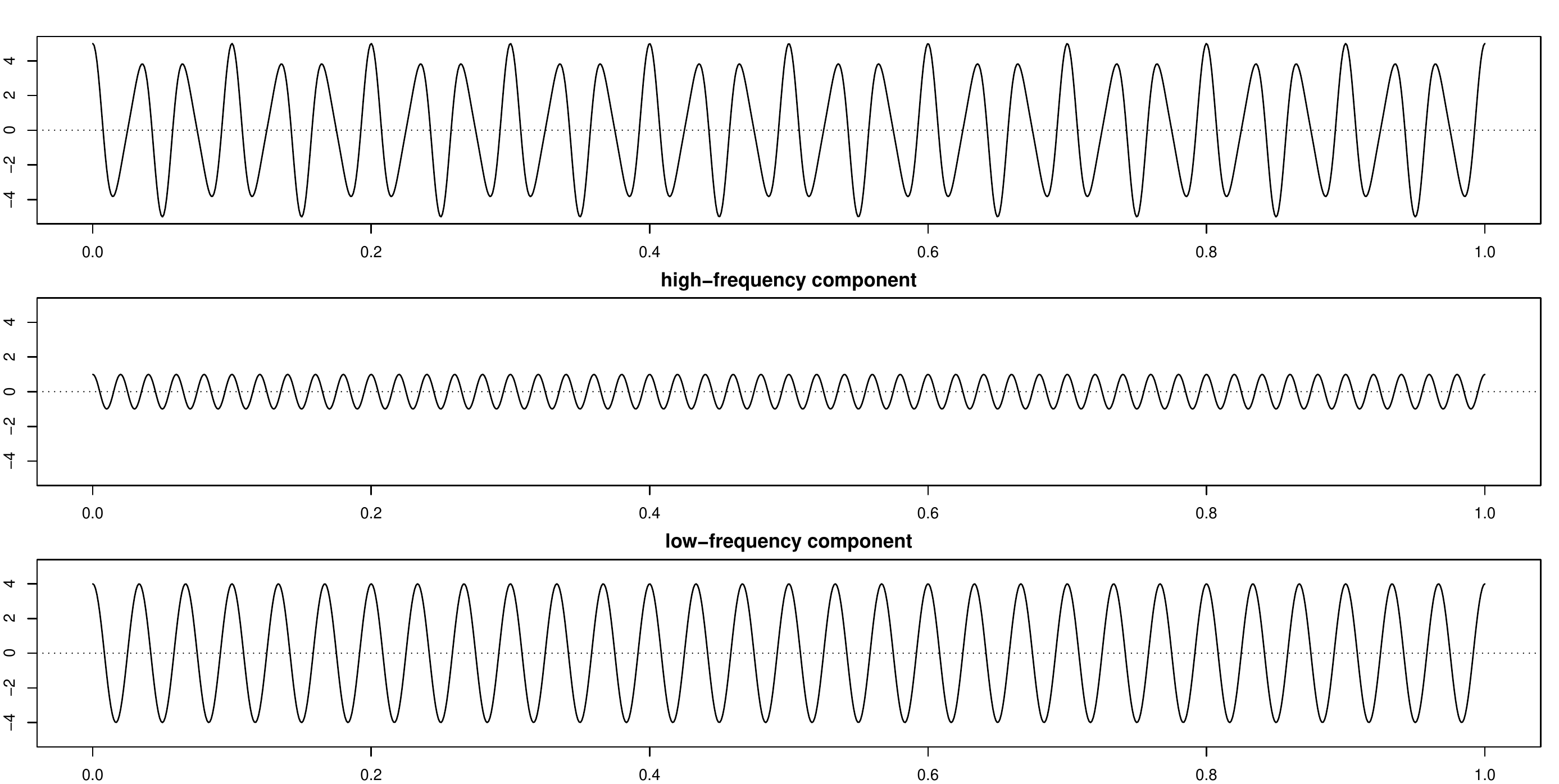}
\caption{A signal of two components, $X_t = \cos(100 \pi t) + 4 \cos(60 \pi t)$.}
\label{signal3}
\end{figure}

\begin{figure}[!t]
\centering
\includegraphics[width=0.95\linewidth]{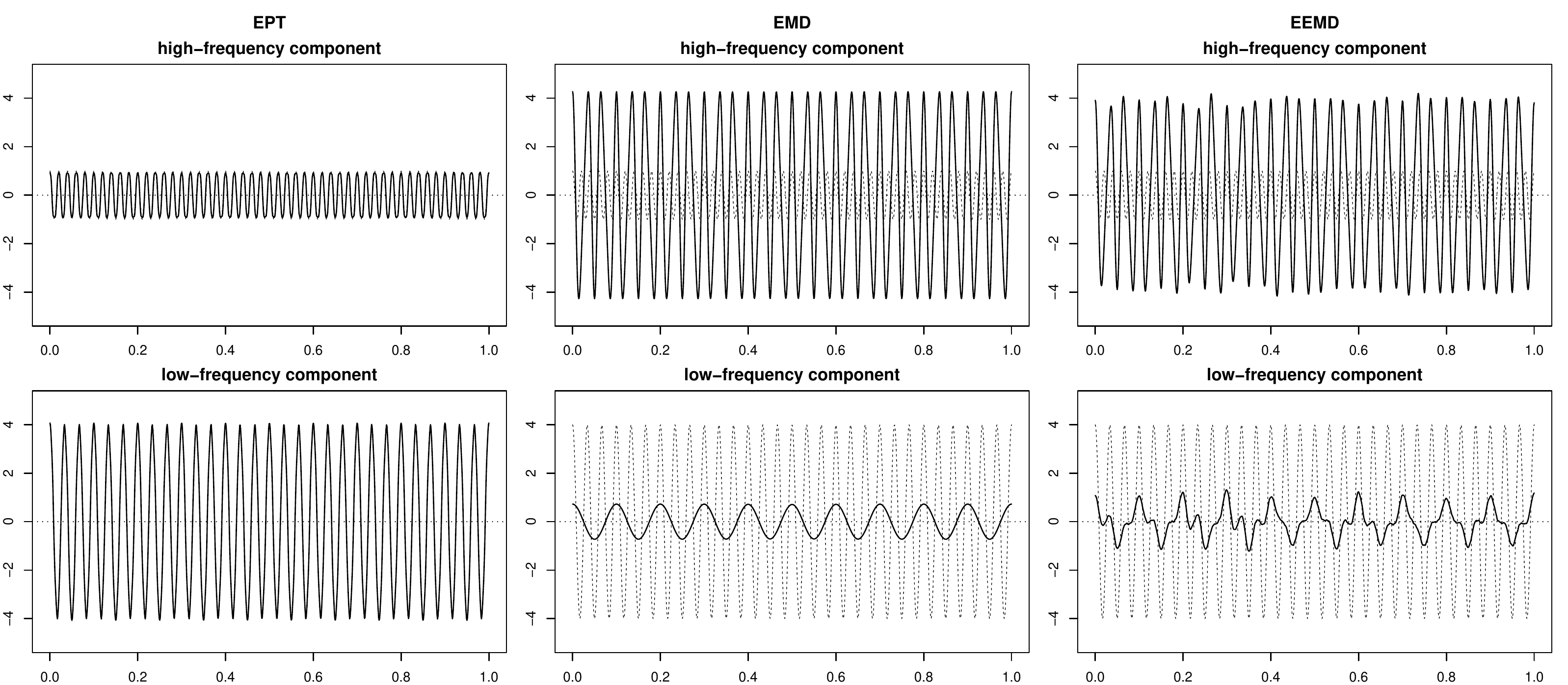}
\caption{Decomposition of signal $X_t = \cos(100 \pi t) + 4 \cos(60 \pi t)$.}
\label{signal3decom}
\end{figure}


The rest of the paper is organized as follows. Section~\ref{sec.EPT} introduces ensemble patch transformation and investigates its utility as a multiscale method. In addition, various statistical measures based on ensemble patch transformation are discussed for data analysis and visualization technique. In Section~\ref{sec.proposal}, a new method for signal decomposition is proposed with a practical algorithm. Furthermore, some theoretical properties of the proposed algorithm are investigated. Section~\ref{sec.simulation} presents simulation studies and a real data example to evaluate empirical performance of the proposed method. In Section 5, as a practical issue of the proposed method, selection of the size parameter is discussed. Lastly, conclusions are addressed in Section~\ref{sec.conclusion}.

Before closing this section, we remark that in literature, there are numerous studies for signal decomposition. \cite{Dragomiretskiy2014} developed variational mode decomposition (VMD) for tone detection and separation of a signal. VMD first conducts discrete Fourier transform for detecting frequency information of each mode, and then identifies several meaningful modes using the detected main frequencies. For this procedure, it is required to preset the number of modes for the decomposition. However, it is difficult to know the number of the meaningful modes according to their frequency information in advance. 
As a data-adaptive procedure, Huang et al. (1998) proposed  empirical mode decomposition (EMD). Due to its robustness to presence of nonlinearity and nonstationarity, EMD has been applied to various fields. Since EMD is based on empirical algorithm, it raises several methodological issues such as identification of local frequency pattern and intermittency. There have been many proposals to enhance the performance of the conventional EMD. Wu and Huang (2009) developed the ensemble EMD (EEMD) taking average of the simulated signals, and its variants have been proposed by several authors. These include the complementary ensemble EMD of Yeh et al. (2010), the complete ensemble EMD with adaptive noise of Torres et al. (2011), and the improved complete ensemble EMD of Colominas et al. (2014). Daubechies  et al. (2011) proposed an alternative method of EMD, termed, synchrosqueezed wavelet transforms, which is based on reassignment methods of wavelet coefficients. Thakur et al. (2013) discussed a selection method of various parameters in the discrete version, and Thakur and Wu (2011) and Meignen et al. (2012)  proposed some methods that are robust to non-uniform samples and noise via synchrosqueezing techniques. 

\section{Ensemble Patch Transform}\label{sec.EPT}

\subsection{Multiscale Patch Transform}
In this section, we introduce a multiscale patch transform of one-dimensional sequence that is designed for processing data and building blocks. We first define a patch process of a real-valued univariate process $(X_t)_{t}$. A patch at the point $(t, X_t)$ is a polygon containing neighbors of the point $(t, X_t)$. The patch is a tool capturing the multiscale characteristics of a signal. A level of multiresolution is controlled by the size of the patch, and various shapes of the patch can be used according to the purpose of data analysis. The patch is formally defined by its shape and size. Let $\mathcal{T}=\{\tau_i\}_i$ be a set of size parameters for patch with a certain shape such as rectangle and oval. For $\tau\in \mathcal{T}$, let $P_t^{\tau}(X_t)$ denote the patch process at the point $(t, X_t)$ that is generated by a certain shape with size parameter $\tau$. We further define a multiscale patch transform $MP_t^{\mathcal{T}}(X_t)$ at the point $(t, X_t)$ that is defined as a sequence of all patches according to various $\tau_i$'s, 
\[
MP_t^{\mathcal{T}}(X_t):=\{P_t^{\tau_i}(X_t)\}_{i=1,\ldots,|\mathcal{T}|}. 
\]
As one can see, the precise definition of $MP_t^{\mathcal{T}}(X_t)$ depends on the shape of the patch. As for typical case, rectangle and oval can be considered as follows. Of course, we can take other shapes as well.   

\noindent {\bf Rectangle patch} : For a given point $(t, X_t)$ and $\tau\in \mathcal{T}$, this patch is centered at the point $(t, X_t)$ and is a closed rectangle formed by the points $(t+k, \min_{k\in[-\tau/2,\tau/2]}\{X_{t+k}\}-0.5 \gamma \tau)$ and 
$(t+k, \max_{k\in[-\tau/2,\tau/2]}\{X_{t+k}\}+0.5 \gamma \tau)$ for $k\in[-\tau/2,\tau/2]$.
For the rectangle patch, the width length is $\tau$ and hight length $h_t^\tau$ is
\[
h_t^\tau=\max_{k\in[-\tau/2,\tau/2]}\{X_{t+k}\}-\min_{k\in[-\tau/2,\tau/2]}\{X_{t+k}\}+\gamma\tau,
\]  
where $\gamma$ is a scale factor. In this study, we set $\gamma=1$. 

\noindent {\bf Oval patch} : For a given point $(t, X_t)$ and $\tau\in \mathcal{T}$, this patch is centered at the point $(t, X_t)$ and is characterized by boundaries  $(t+k, X_{t+k}\pm\gamma\sqrt{\tau^2/4-k^2})$, $k\in[-\tau/2,\tau/2]$ where $\gamma$ is a scale factor. The width length for the oval patch is $\tau$ as for the rectangle patch, and the height is of decreasing pattern as moving away from a given point $(t, X_t)$.

\begin{figure}[!t]
\centering
\includegraphics[width=0.75\linewidth]{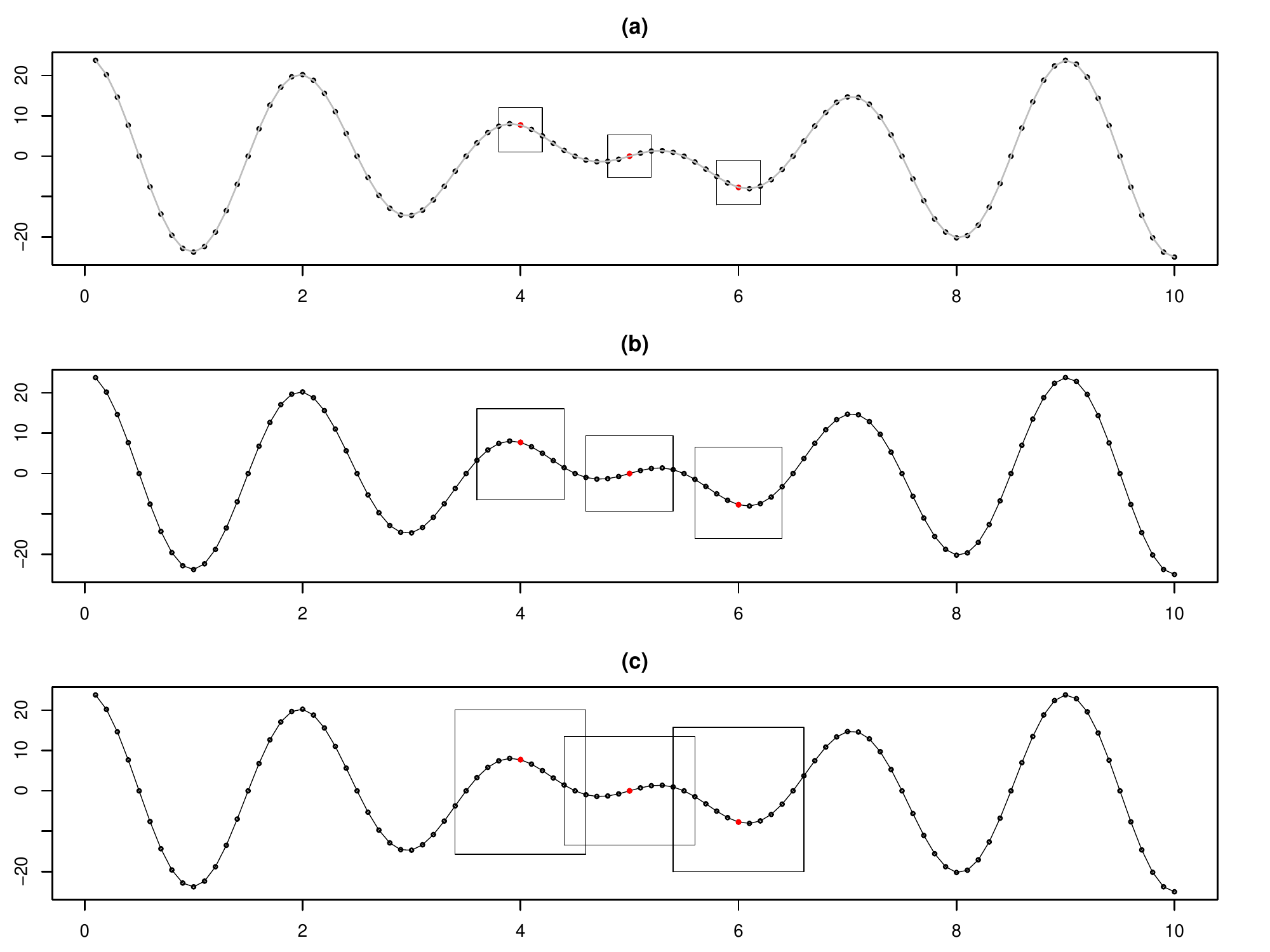}
\vspace{-3mm}
\caption{Patches with rectangle shape $P_{t_i}^{\tau}(X_t)$ of signal $X_t=25\cos(0.1\pi t)\cos(\pi t)$ at center points of the patches, $t_i=4,5,6$. (a) $\tau=4$, (b) $\tau=8$ and (c) $\tau=12$.} 
\label{patch}
\end{figure}

\begin{figure}[!t]
\centering
\includegraphics[width=0.75\linewidth]{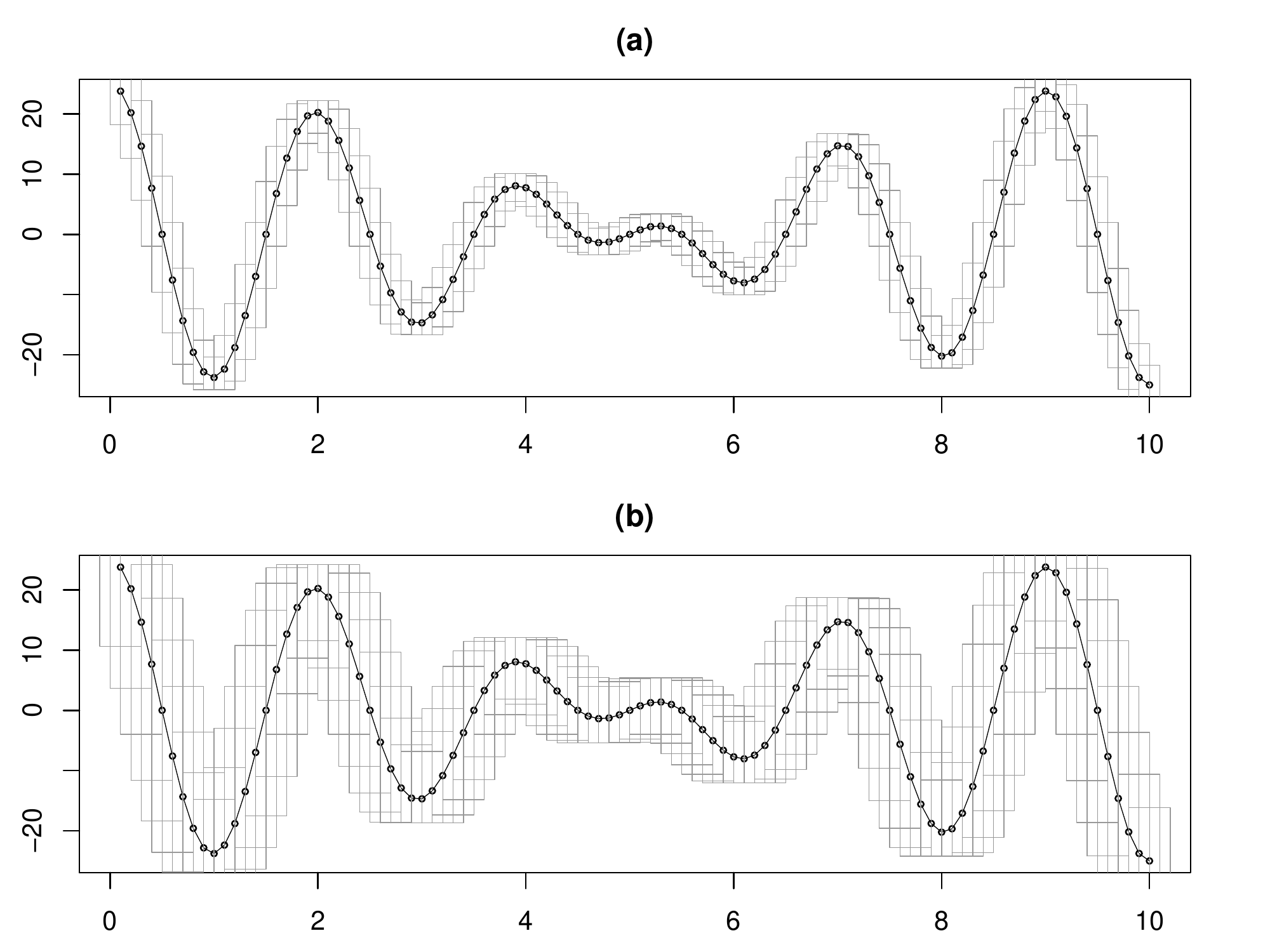}
\vspace{-3mm}
\caption{Patches with rectangle shape $P_{t_i}^{\tau}(X_t)$ of signal $X_t=25\cos(0.1\pi t)\cos(\pi t)$, $0\le t\le 10$. (a) $\tau=2$ and (b) $\tau=4$} 
\label{patchentire}
\end{figure}

For illustration of the patch process, we consider a deterministic signal $X_t=25\cos(0.1\pi t)\cos(\pi t)$, $0\le t\le 10$. We then obtain a  sequence $\{X_{t_i}\}_{i=1}^{100}$ with $t_i=iT$ and sampling rate $T=1/10$ from the continuous signal $X_t$.  Figure~\ref{patch} shows rectangle patches $P_{t_i}^{\tau}(X_t)$ of the sequence $\{X_{t_i}\}$ that are respectively performed at certain time points $t_i=4,5,6$ marked by red dots. We consider three different size parameters $\tau=4,8,12$ for generating patches, and obtain a multiscale patch $MP_t^{\mathcal{T}}(X_t)$ by combining the three patches in Figure~\ref{patch}(a)--(c). Figure~\ref{patchentire} shows patches in the entire time domain with the parameters $\tau=2$ and 4, respectively.  

From Figures~\ref{patch},~\ref{patchentire} and the definitions, the patch at a particular time point $t$ is an object that contains multiple observations around the time point $t$; thus, for further statistical analysis, it is necessary to use some statistics that summarize informations of $P_t^{\tau}(X_t)$ and $MP_t^{\mathcal{T}}(X_t)$. For this purpose, we consider a measure $\mathcal{K}(P_t^{\tau}(X_t))$ that produces a single statistic at time point $t$. Some possible measures $\mathcal{K}(\cdot)$ are two-fold: one is for central tendency and the other is for dispersion. As for measures for central tendency, in this study, we present the following two measures. Suppose that we obtain the patch $P_t^{\tau}(X_t)$ for a fixed $\tau$. 
\begin{itemize}
\item $\mbox{Ave}_t^\tau(X_t)= \mbox{average}(\{X_{t_i}\})$, where $\{X_{t_i}\}$ denote observations in the patch $P_t^{\tau}(X_t)$. 

\item $M_t^\tau(X_t)=\frac{1}{2}(L_t^\tau(X_t)+U_t^\tau(X_t))$, where $L_t^\tau(X_t)$ and $U_t^\tau(X_t)$ denote lower and upper envelopes of the patch $P_t^{\tau}(X_t)$, respectively. $M_t^\tau(X_t)$ is called mean envelope. The lower envelope $L_t^\tau(X_t)$ and upper envelope $U_t^\tau(X_t)$ of the rectangle patch are 
$$
L_t^\tau(X_t) = \min_{k\in[-\tau/2,\tau/2]}\{X_{t+k}\}-0.5 \gamma \tau, \quad
U_t^\tau(X_t) = \max_{k\in[-\tau/2,\tau/2]}\{X_{t+k}\}+0.5 \gamma \tau.
$$
The lower envelope $L_t^\tau(X_t)$ and upper envelope $U_t^\tau(X_t)$ of the oval patch are 
$$
L_t^\tau(X_t) = \min_{k\in[-\tau/2,\tau/2]} \{X_{t+k} - \gamma\sqrt{\tau^2/4-k^2} \}, \quad
U_t^\tau(X_t) = \max_{k\in[-\tau/2,\tau/2]} \{X_{t+k} + \gamma\sqrt{\tau^2/4-k^2} \}.
$$
\end{itemize}
For dispersion measure, we consider the followings
\begin{itemize}
\item $\mbox{sd}_\tau(X_t)= \sqrt{\mbox{Var}(\{X_{t_i}\})}$. 

\item $R_t^\tau(X_t)=U_t^\tau(X_t)-L_t^\tau(X_t)$. 
\end{itemize}

\begin{figure}[!t]
\centering
\includegraphics[width=0.82\linewidth]{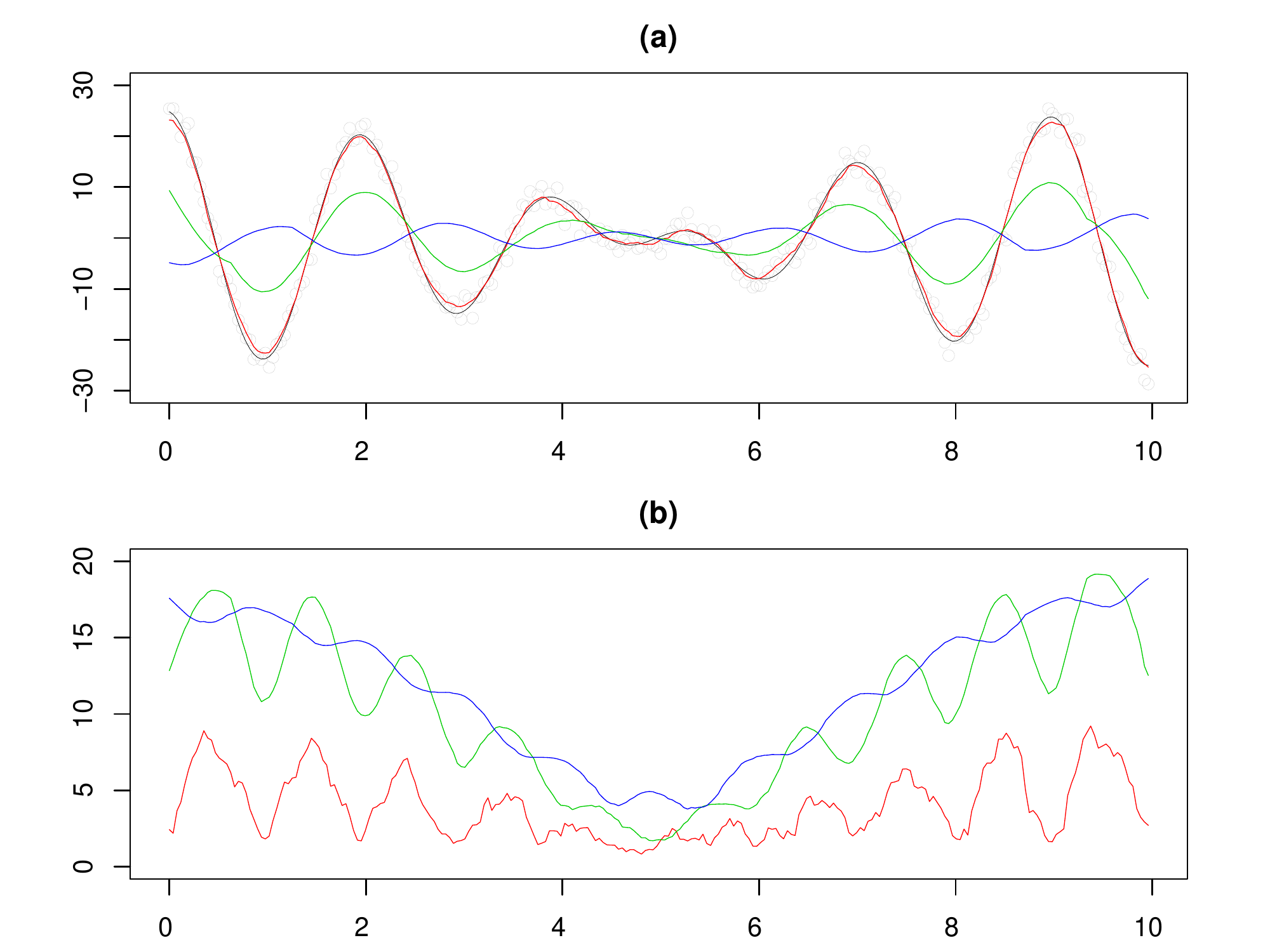}
\vspace{-3mm}
\caption{(a) Noisy sequence (open circles), true function (black line), $\mbox{Ave}_t^\tau(X_t)$ by $\tau=8$ (red line), $\tau=32$ (green line), $\tau=64$ (blue line), and (b) $\mbox{sd}_t^\tau(X_t)$ by $\tau=8$ (red line), $\tau=32$ (green line), $\tau=64$ (blue line).}
\label{patchstat}
\end{figure}

Figure~\ref{patchstat} shows $\mbox{Ave}_t^\tau(X_t)$ and $\mbox{sd}_t^\tau(X_t)$ with size parameters $\tau=8,32,64$ for a noisy signal  $X_t=25\cos(0.1\pi t)\cos(\pi t)+\sigma\epsilon_t$, where $\sigma=1.8$ and $\epsilon_t$ denote i.i.d. standard Gaussian random variables. As the value of size parameter $\tau$ increases, a central measure $\mbox{Ave}_t^\tau(X_t)$ is getting smoother with representing the global trend of the observations. On the other hand, the values of $\mbox{sd}_t^\tau(X_t)$ at both boundaries are large, compared to those at center over all $\tau$'s, and $\mbox{sd}_t^\tau(X_t)$ becomes larger as $\tau$ increases since a large patch contains more observations. Further, it seems that $\mbox{sd}_t^\tau(X_t)$ by $\tau=64$ is capable of identifying the temporal variability of the signal well. 

\subsection{Ensemble Patch Transform}
To improve the temporal resolution of the patch and its measures, we introduce an ensemble patch process of a real-valued univariate process $(X_t)_{t}$. 
\begin{dfn}
Let $(X_t)_t$ be a real-valued univariate process. Let $\mathcal{T}$ denote a set of size parameters for the patch. For any $\tau\in \mathcal{T}$, the $\ell$th shifted patch at time point $t$ is defined as $P_{t+\ell}^{\tau}(X_t)$, $\ell\in[-\tau/2,\tau/2]$. Then, a fixed $\tau\in \mathcal{T}$, a collection of all possible shifted patches at time point $t$ is defined as ensemble patch,
\[
EP_t^\tau(X_t):= \left\{P_{t+\ell}^{\tau}(X_t) : \ell\in[-\tau/2,\tau/2] \right\}. 
\]
Finally, as a dictionary, the multiscale ensemble patch process is defined the sequence of all sets of $EP_t^\tau(X_t)$ over various $\tau$'s as
\[
MEP_t^{\mathcal{T}}(X_t):= \left\{ EP_t^\tau(X_t) : \tau\in\mathcal{T} \right\}. 
\] 
\end{dfn}

\begin{figure}[!t]
\centering
\includegraphics[width=0.82\linewidth]{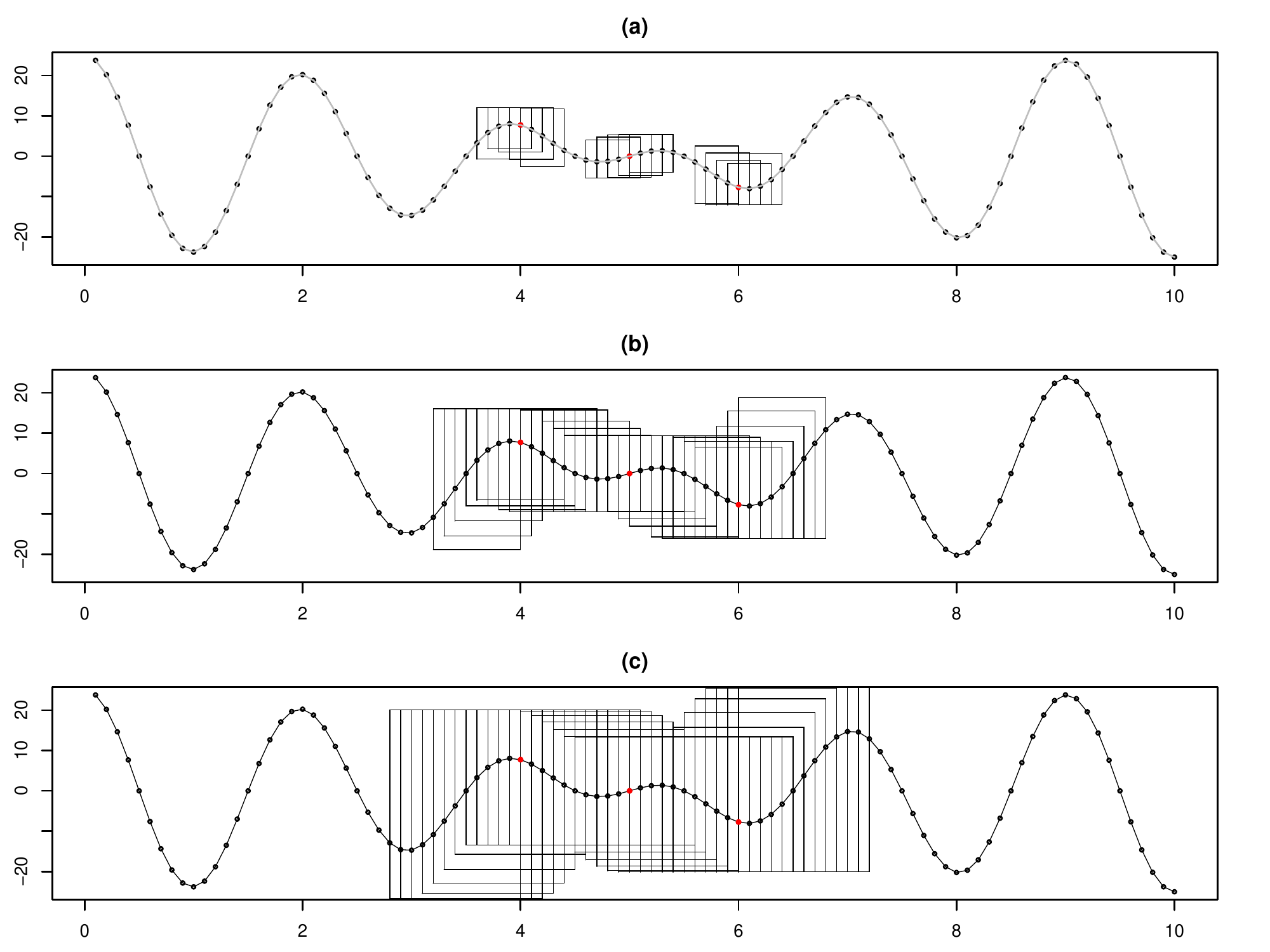}
\vspace{-3mm}
\caption{Ensemble patches with rectangle shape $EP_{t_i}^{\tau}(X_t)$ of signal $X_t=25\cos(0.1\pi t)\cos(\pi t)$ at center points of the patches, $t_i=4,5,6$. (a) $\tau=4$, (b) $ \tau=8$ and (c) $\tau=12$.} 
\label{epatch}
\end{figure}

For the sequence $\{X_{t_i}\}$ in Figure~\ref{patch}, we generate ensemble rectangle patches at the same time points $t_i=4,5,6$ according to size parameters $\tau=4,8,12$, which are displayed in Figure~\ref{epatch}. A multiscale ensemble patch $MEP_t^{\mathcal{T}}(X_t)$ is obtained by combining the three ensemble patches in Figure~\ref{epatch}(a)--(c). 

Similarly, for further data analysis, we need to consider some statistics of ensemble patch $EP_t^\tau(X_t)$. We first consider a measure of each shifted patch $\mathcal{K}(P_{t+\ell}^{\tau}(X_t))$ and then obtain an ensemble measure by averaging  $\mathcal{K}(P_{t+\ell}^{\tau}(X_t))$'s over $\ell$ in $[-\tau/2,\tau/2]$.  More specifically, we obtain the following ensemble measures for central tendency and dispersion: For a fixed $\tau$, suppose that we obtain the  collection of all shifted patches at time point $t$, $EP_t^\tau(X_t)$ of the 
patch $P_t^{\tau}(X_t)$. 
\begin{itemize}
\item $\mbox{EAve}_t^\tau(X_t)= \mbox{average}(\mbox{Ave}_{t+\ell}^{\tau}(X_{t}))$ over $\ell$'s, where $\mbox{Ave}_{t+\ell}^{\tau}(X_{t})$ denotes the simple average of observations in the shifted patch $P_{t+\ell}^{\tau}(X_t)$. 

\item $EM_t^\tau(X_t)=\mbox{average}(M_{t+\ell}^{\tau}(X_t))$ over $\ell$'s, where $M_{t+\ell}^{\tau}(X_t)$ denotes the average of $L_{t+\ell}^{\tau}(X_t)$ and $U_{t+\ell}^{\tau}(X_t)$ that are lower and upper envelopes of the patch $P_{t+\ell}^{\tau}(X_t)$.
\end{itemize}

\begin{itemize}
\item $\mbox{Esd}_t^\tau(X_t)= \mbox{average}(\mbox{sd}_{t+\ell}^{\tau}(X_t))$ over $\ell$'s, where $\mbox{sd}_{t+\ell}^{\tau}(X_t)$ denotes the standard deviation of observations in the shifted patch $P_{t+\ell}^{\tau}(X_t)$. 

\item $ER_t^\tau(X_t)= \mbox{average}(R_{t+\ell}^{\tau}(X_t))$ over $\ell$'s, where $R_{t+\ell}^{\tau}(X_t)=U_{t+\ell}^{\tau}(X_t)-L_{t+\ell}^{\tau}(X_t)$. 
\end{itemize}

\begin{figure}[!t]
\centering
\includegraphics[width=0.82\linewidth]{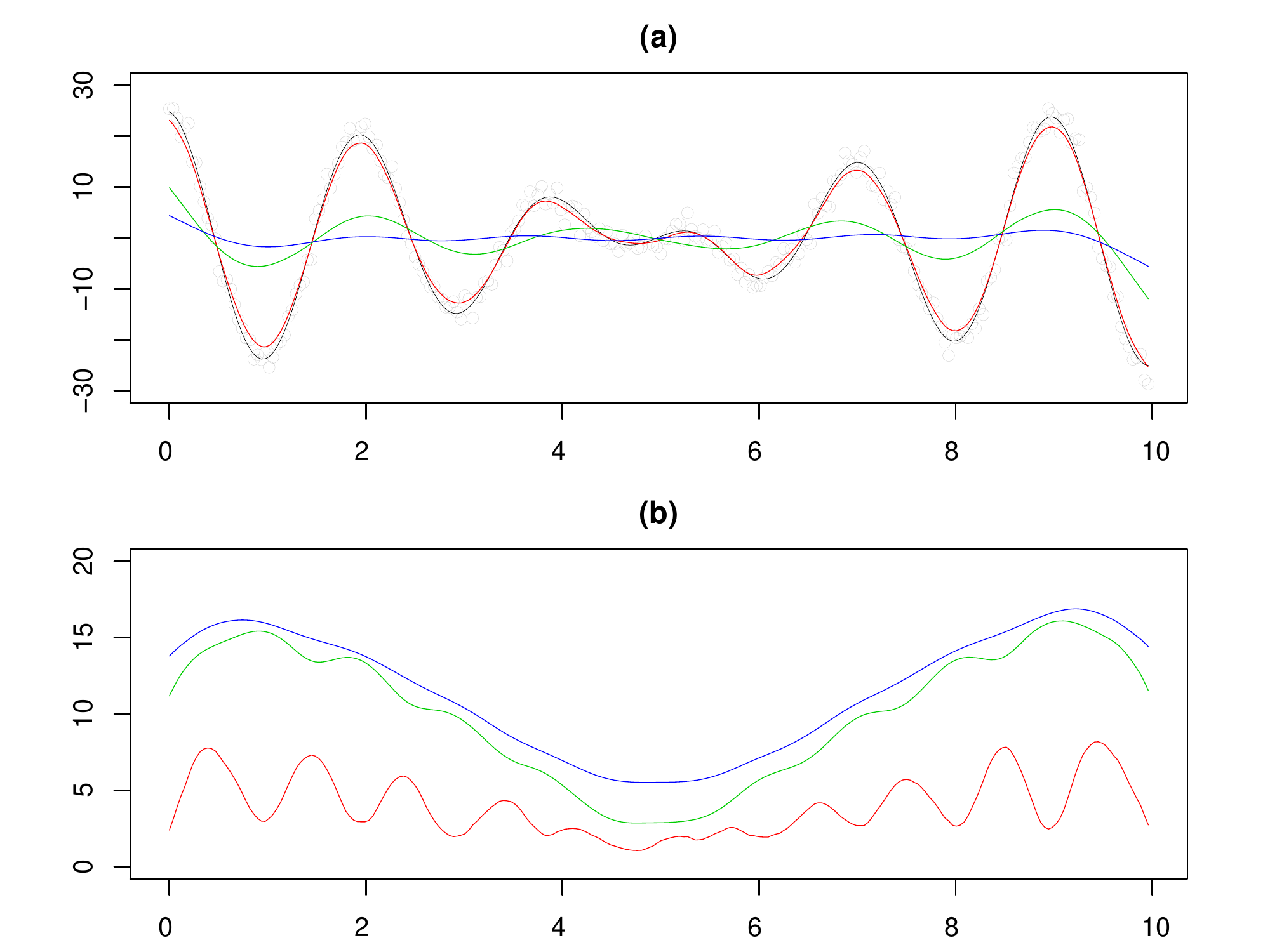}
\vspace{-3mm}
\caption{(a) Noisy sequence (open circles), true function (black line), $\mbox{EAve}_t^\tau(X_t)$ by $\tau=8$ (red line), $\tau=32$ (green line), $\tau=64$ (blue line), and (b) $\mbox{Esd}_t^\tau(X_t)$ by $\tau=8$ (red line), $\tau=32$ (green line), $\tau=64$ (blue line).}
\label{epatchstat}
\end{figure}

We obtain some measures based on ensemble patches of the noisy signal in Figure~\ref{patchstat}, $\mbox{EAve}_t^\tau(X_t)$ and $\mbox{Esd}_t^\tau(X_t)$ with size parameters $\tau=8,32,64$, which are shown in Figure~\ref{epatchstat}. As the value of $\tau$ increases, the central measure $\mbox{EAve}_t^\tau(X_t)$ is getting smoother, and  the dispersion measure $\mbox{Esd}_t^\tau(X_t)$  becomes larger and the values are relatively large at both boundaries. Furthermore, by comparison of the ensemble results with the single patch results in Figure~\ref{patchstat}, we have some observations: (a) The central measure by ensemble patches represents the temporal trend of the underlying function well, compared to that by single patches.  (b) The dispersion measure with large $\tau$ by ensemble patches identifies a local variability of the underlying function efficiently. (c) The temporal resolution of both measures by ensemble patches are much finer than those of single patches. In addition,  ensemble patches are able to obtain various statistics that are adapted for a purpose of data analysis. For example, as an alternative central measure, we can consider median for each patch $P_{t+\ell}^\tau(X_t)$, say Med$_{t+\ell}^{\tau}(X_t)$  and the corresponding mean of Med$_{t+\ell}^{\tau}(X_t)$ over $\ell$, EMed$_t^{\tau}(X_t)$.

We finally remark that thick-pen transformation by Fryzlewicz and Oh (2011) is a special case of the above $MEP_t^{\mathcal{T}}(X_t)$ with $L_{t+\ell}^{\tau}(X_t)$ and $U_{t+\ell}^{\tau}(X_t)$ at the shifting index $\ell=0$.

\subsection{Visualization}
The proposed transform holds inherently multiscale features owing to the parameter $\tau$ that plays a role in controlling the size of patch. That is, the size parameter of patch acts as the scale parameter of multiscale features. Scale-space concept might provide a view-point on visualization of data, which considers a family of representations of data indexed by the scale parameter $\tau$ instead of the conventional dot-connected plot. 

Here we present multiscale visualization techniques based on patch transform $P_t^\tau(X_t)$ and ensemble patch transform $EP_t^\tau(X_t)$, termed C-map (centrality map) and D-map (dispersion map). C-map and D-map are time-scale representation of two-dimensional array (matrix) with the $(t,\tau)$ element as $\mbox{Ave}_t^\tau(X_t)$ (or $\mbox{EAve}_t^\tau(X_t)$)  and $\mbox{sd}_t^\tau(X_t)$ (or $\mbox{Esd}_t^\tau(X_t)$), where  $\mbox{Ave}_t^\tau(X_t)$ and $\mbox{sd}_t^\tau(X_t)$ represent centrality and dispersion measures of data in patch, respectively. 

\begin{figure}[!t]
\centering
\includegraphics[width=0.75\linewidth]{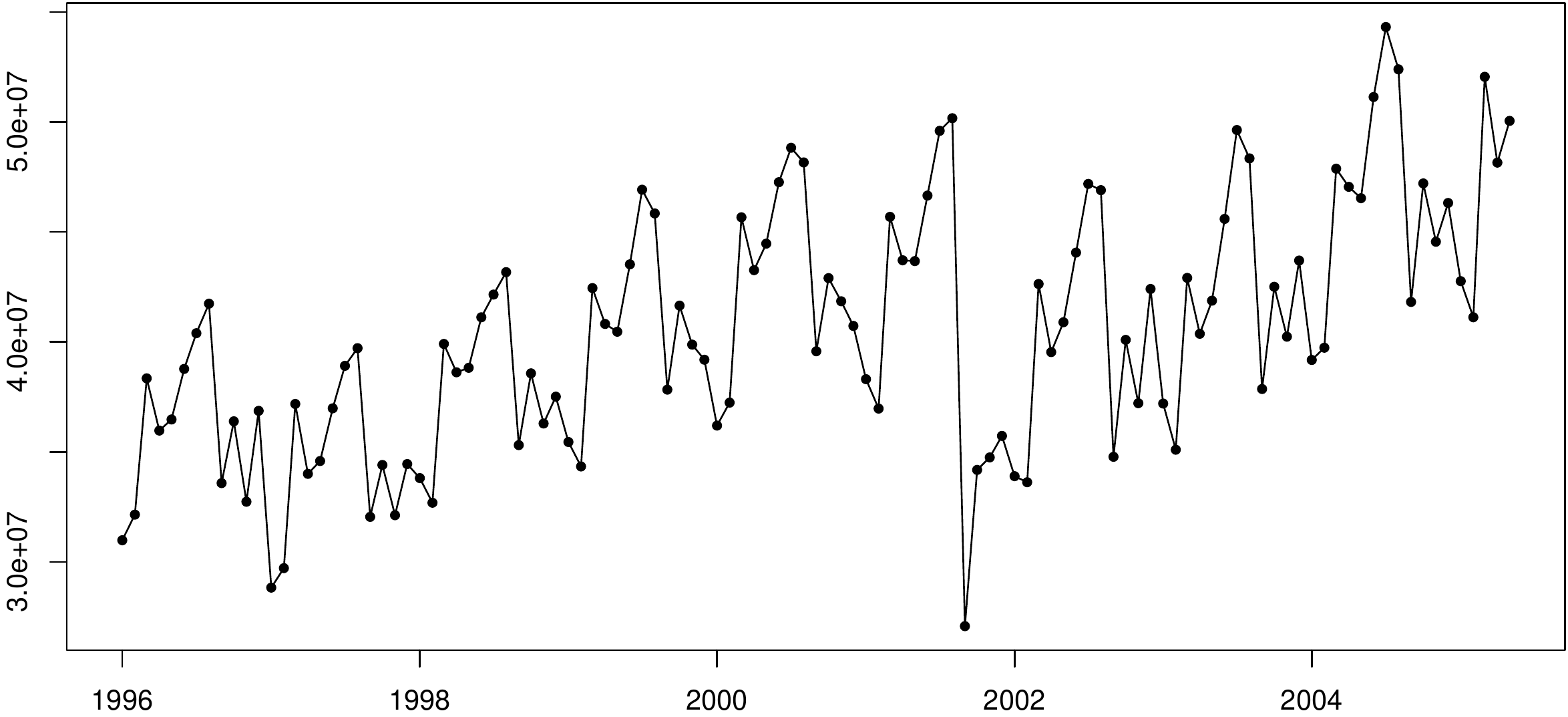}
\vspace{-3mm}
\caption{Airmile data}
\label{airmile}
\end{figure}

For illustration of the maps, we use monthly airline passenger-mile data in the United States during the period from January 1996 to May 2005 (Cryer and Chan, 2008), which is shown in Figure~\ref{airmile}. The data show a strong seasonality with holidays effects, and these are increasing linearly overall with an intervention in September 2001 and several months thereafter due to the terrorist acts on September 11, 2001. To understand the time-varying structure of air passenger-miles, we construct  C-map and D-map of the data based on the measures $\mbox{Ave}_t^\tau(X_t)$ (or $\mbox{EAve}_t^\tau(X_t)$)  and $\mbox{sd}_t^\tau(X_t)$ (or $\mbox{Esd}_t^\tau(X_t)$) shown in Figure~\ref{map}.  From the centrality map of panels (a) and (c), we observe the temporal patterns of the data over the domain, which are increasing with a sudden drop around September 2001. As expected, ensemble C-map in panel (c) based on $\mbox{EAve}_t^\tau(X_t)$ holds an enhanced temporal resolution, compared to C-map of $\mbox{Ave}_t^\tau(X_t)$ in panel (a). From two dispersion D-maps in panels (b) and (d) based on $\mbox{sd}_t^\tau(X_t)$ and $\mbox{Esd}_t^\tau(X_t)$, it is capable of identifying the dependent structure of the data evolving on time; thus, the sudden drop of air passenger-miles near September 2001 can be easily detected.   

\begin{figure}[!t] 
\centering
\includegraphics[width=0.8\linewidth]{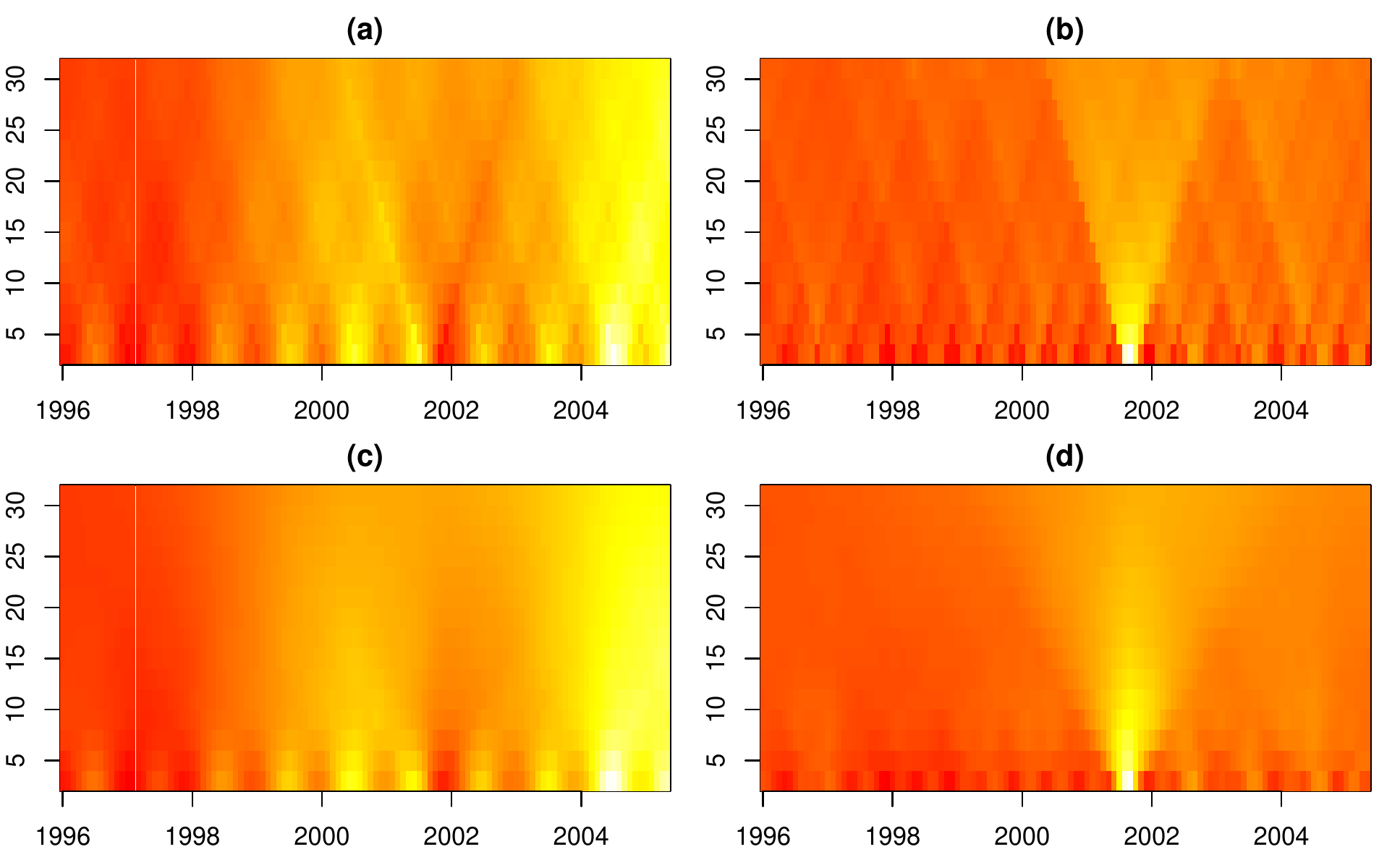}
\vspace{-3mm}
\caption{C-maps and D-maps of airmile data. (a) C-map by $\mbox{Ave}_t^\tau(X_t)$, (b) D-map by $\mbox{sd}_t^\tau(X_t)$, (c) C-map by $\mbox{EAve}_t^\tau(X_t)$, and (d) D-map by $\mbox{Esd}_t^\tau(X_t)$.}
\label{map}
\end{figure}

Moreover, we construct derivatives of C-map and D-map with respect to time $t$ and scale $\tau$, termed DC-map and DD-map that are defined as two-dimensional array with the ($t,\tau$) element as $\Delta\mbox{EAve}_t^\tau(X_t)/\Delta t$ ($\Delta\mbox{EAve}_t^\tau(X_t)/\Delta \tau$) and $\Delta\mbox{Esd}_t^\tau(X_t)/\Delta t$ ($\Delta\mbox{Esd}_t^\tau(X_t)/\Delta \tau$). Figure~\ref{dmap} shows DC-maps and DD-maps of the air passenger-mile data with respect to time $t$ and scale $\tau$. Especially, the DC-map and DD-map with respect to time in panels (a) and (c) detect the intervention clearly.  

\begin{figure}[!t]
\centering
\includegraphics[width=0.8\linewidth]{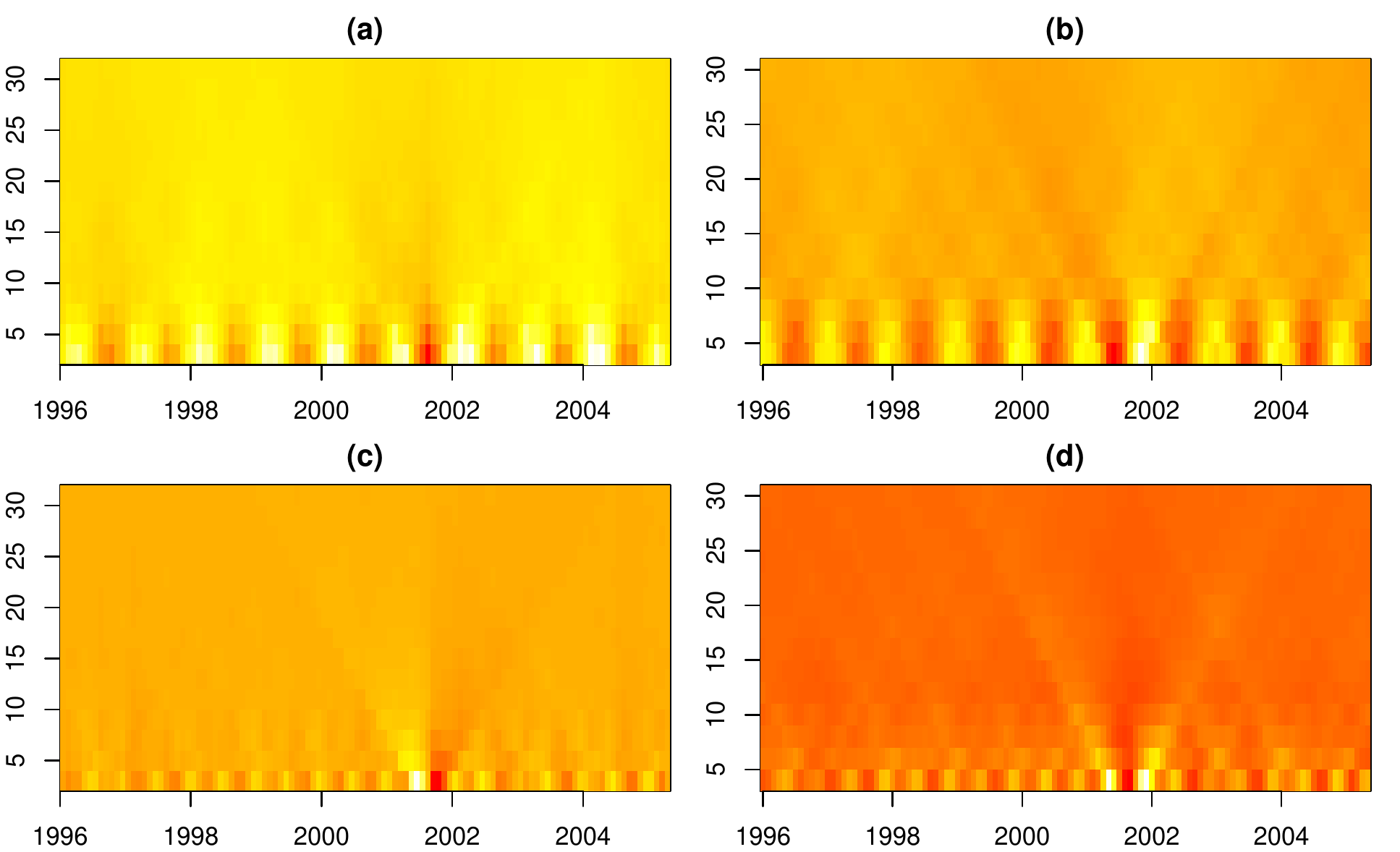}
\vspace{-3mm}
\caption{DC-maps and DD-maps of the air passenger-mile data. (a) DC-map by $\Delta\mbox{EAve}_t^\tau(X_t)/\Delta t$, (b) DC-map by $\Delta\mbox{EAve}_t^\tau(X_t)/\Delta \tau$, (c) DD-map by $\Delta\mbox{Esd}_t^\tau(X_t)/\Delta t$, and (d) DD-map by $\Delta\mbox{Esd}_t^\tau(X_t)/\Delta \tau$.}
\label{dmap}
\end{figure} 

\section{Ensemble Patch Filtering and Decomposition} \label{sec.proposal} 

\subsection{Ensemble Patch Filtering}
When a signal consists of several components with their own frequencies, ensemble patch transformation can be utilized as a low-pass or a high-pass filter. Figure~\ref{ptensemble} illustrates filtering process of the ensemble mean envelope. The top panel shows a sinusoidal signal $X_t=\cos(50 \pi t)+\cos(10 \pi t)+2t \ (t \in [0.35, 0.55])$, and depicts three shifted rectangle patches covering a point $X_t$ at $t=0.45$ of open circle. Each  $\ell$th shifted patch produces upper and lower envelopes $U_{t+\ell}^{\tau}(X_t)$ and $L_{t+\ell}^{\tau}(X_t)$ at time point $t$. The black dots denote the mean envelope $M_{t+\ell}^{\tau}(X_t)$ at time point $t=0.45$ for the shifted patches a, b and c. Similarly, with shifting the patch over the entire time domain, we construct a mean envelope for each shifted patch. The bottom panel of Figure~\ref{ptensemble} shows three mean envelopes (dotted line), respectively.  Furthermore, for improvement of the stability, we take an ensemble average of three mean envelopes, which results in the ensemble mean envelope marked by solid line.  It seems that the ensemble mean envelope represents a lower frequency component of the signal. We note that, although we use only three shifted mean envelops for illustration purpose, the possible number of shifted mean envelopes for a given point is the same as the size parameter $\tau$ of patch generally. The ensemble mean envelope might provide more stable result.

\begin{figure}[!t]
\centerfig{0.85 \columnwidth}{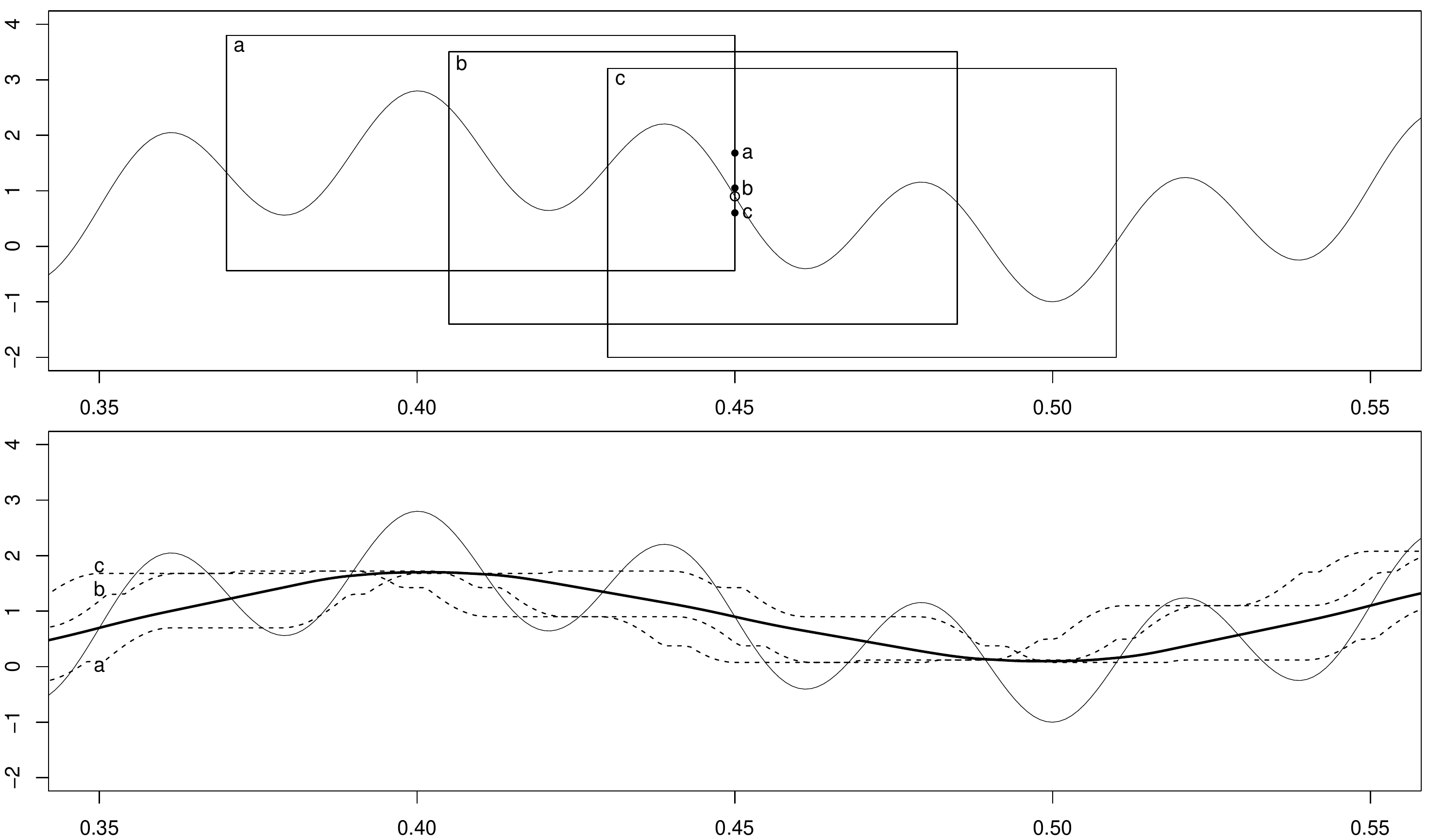}
\caption{Ensemble mean envelope for a signal $X_t$.}
\label{ptensemble}
\end{figure}

For demonstrating a utility of this ensemble approach, we consider a synthetic example. Figure~\ref{ptidea} shows a signal $X_t=\cos(50 \pi t)+\cos(10 \pi t)+2t~ (t \in [0, 1])$ in white color and its ensemble patch transformation of rectangle patch with size parameters $\tau=$ 20, 40, 80, 120, 200 and 240, respectively. The lower and upper envelopes $EL_{t}^{\tau}(X_t)$, $EU_{t}^{\tau}(X_t)$ and mean envelope $EM_t^\tau(X_t)$ are obtained by the ensemble approach. The area covered by two envelopes is colored in gray, and mean envelope is denoted by solid lines.
\begin{figure}[!t]
\centerfig{0.95  \columnwidth}{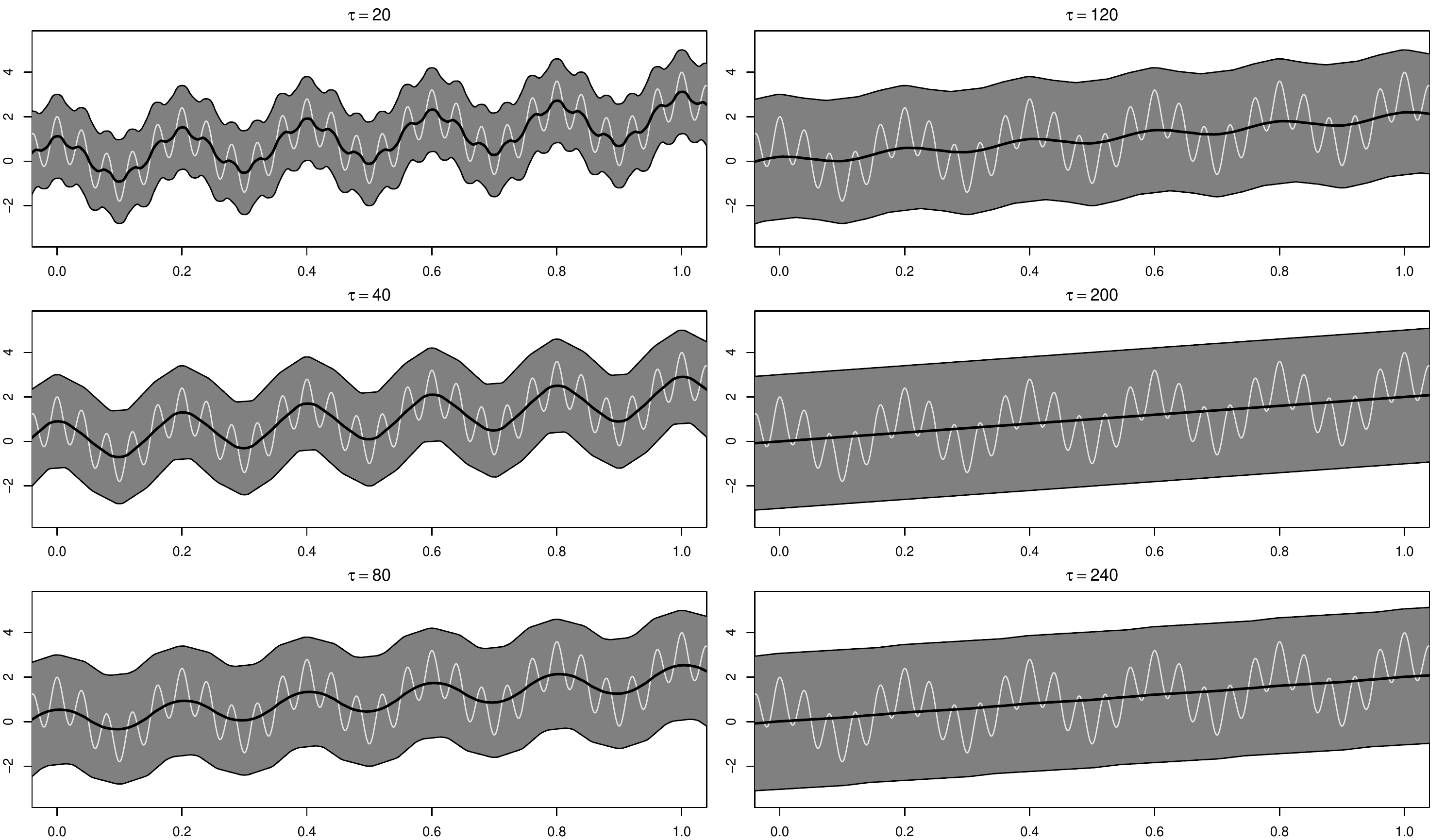}
\caption{ $EM_t^\tau(X_t)$'s of signal $X_t$ according to various $\tau$'s.}
\label{ptidea}
\end{figure}
We observe that with the size parameter $\tau=40$, the ensemble mean envelope suppresses a high-frequency component $\cos(50 \pi t)$. When the size parameter $\tau$ is larger than $200$, both the oscillating patterns of components $\cos(50 \pi t)$ and $\cos(10 \pi t)$ are painted over by patch transformation. As the size parameter $\tau$ is larger, the ensemble mean envelope suppresses the oscillating local pattern, and at the same time represents the lower-frequency pattern. The ensemble mean envelope removes the frequency pattern whose period is less than $\tau$. By controlling the size parameter, mean envelope $EM_t^\tau(X_t)$ of the ensemble patch transformation is implemented as low-pass filter or high-pass filter.

In addition, we perform the same experiment with measure $\mbox{EAve}_t^\tau(X_t)$ that is average of $\mbox{Ave}_{t+\ell}^{\tau}(X_t)$ obtained over ensemble patches. The results  $\mbox{EAve}_t^\tau(X_t)$ (solid line) with different $\tau=$ 20, 40, 80, 120, 200 and 240 are displayed in Figure~\ref{ptidea2}. As one can see, the results are almost identical to those of $EM_t^\tau(X_t)$. 

\begin{figure}[!h]
\centerfig{0.95  \columnwidth}{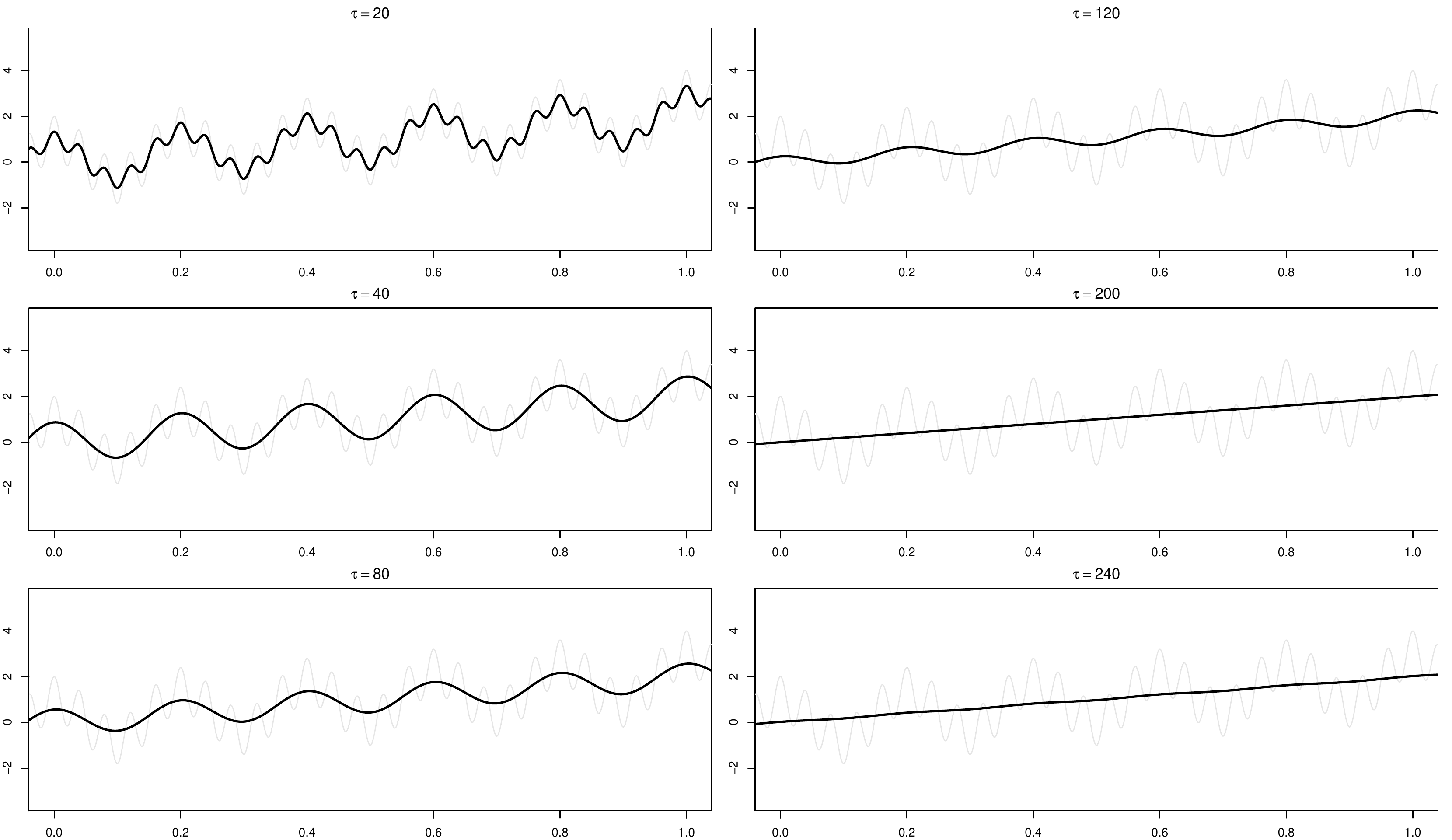}
\caption{ $\mbox{EAve}_t^\tau(X_t)$'s of signal $X_t$ according to various $\tau$'s.}
\label{ptidea2}
\end{figure}

\subsection{Decomposition by Ensemble Patch Filtering}

By adapting the above notion of filter, we would like to decompose a signal into a high-frequency component and a low-frequency residue component. 
We consider a signal
$
X_t = \cos(90 \pi t) + \cos(10 \pi t), \ t \in [0,1],
$
shown in Figure~\ref{signal1}. 
\begin{figure}[!t] 
\centerfig{0.85 \columnwidth}{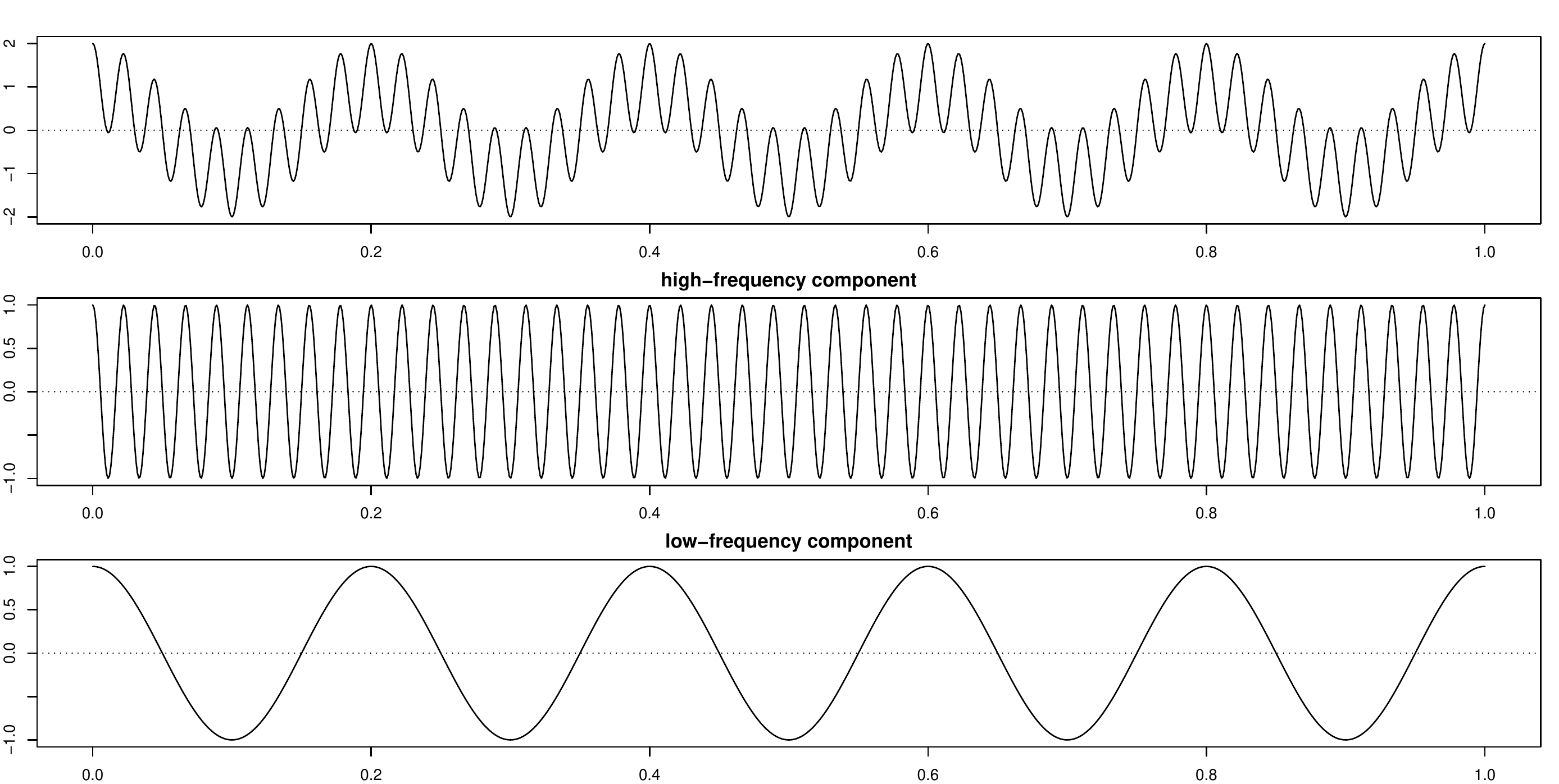}
\caption{A signal $X_t=\cos(90 \pi t) + \cos(10 \pi t)$ and its two components.}
\label{signal1}
\end{figure}
 A snapshot of decomposition procedure by ensemble patch filtering is displayed in Figure~\ref{ptdecom}. From top to down and left to right panels, the first panel illustrates a low-frequency mode, say, $LF_1$ that is an ensemble mean envelope of $X_t$ obtained by $EP_t^\tau(X_t)$ for a given $\tau$, and the corresponding high-frequency mode $HF_1 =X-LF_1$ in the next panel. As one can see, there still exists apparent low-frequency mode in $HF_1$. The third panel shows an ensemble mean envelope of $HF_1$, say, $LF_2$ which seemingly identifies the low-frequency mode of $HF_1$ in the second panel. Now a new high-frequency mode $HF_2 =HF_1-LF_2=X-LF_1-LF_2$ is obtained. In the next iteration, a further ensemble mean envelope of $HF_2$, say, $LF_3$ is almost constant; hence, the corresponding high-frequency mode $HF_3=HF_2-LF_2=X-LF_1-LF_2-LF_3$ represents the true high-frequency component well. Thus, an iterative procedure is required. We note that this iterative process is along with the line of the sifting process of EMD.

\begin{figure}[!t]
\centerfig{0.95\columnwidth}{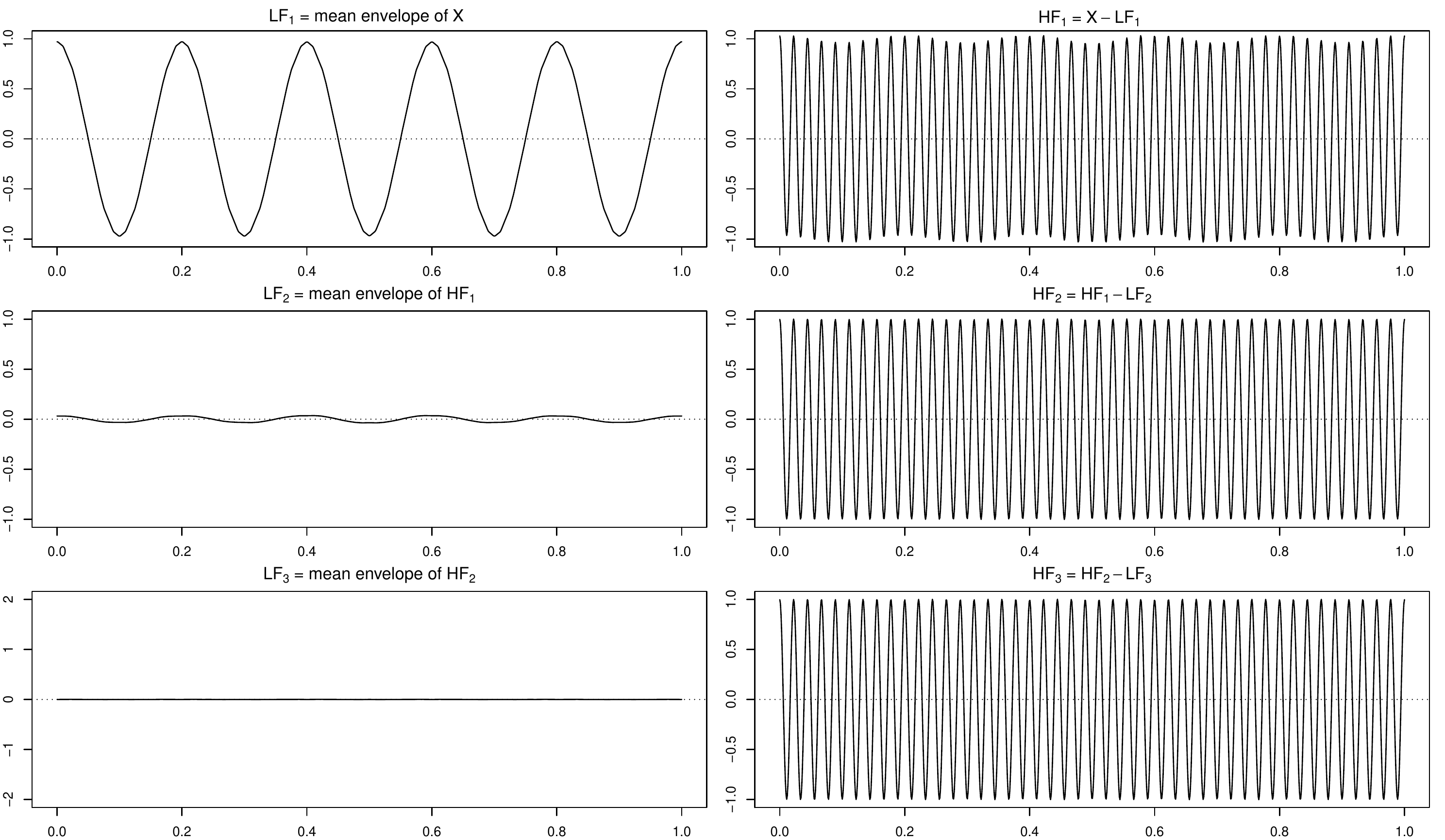}
\caption{Iterative decomposition procedure.}
\label{ptdecom}
\end{figure}

From the above discussion, we propose a practical decomposition algorithm based on ensemble patch filtering. Let $\mathcal{G}_t^\tau(X_t)$ be a generic central measure of $\{P^\tau_{t+\ell}(X_t)\}_\ell$, where $P^\tau_{t+\ell}(X_t)$ is the $\ell$th shifted patch at time $t$ for a given $\tau$.  Suppose that a signal $X_t$ consists of a high-frequency component $h_t$ and a low-frequency component $g_t$ as  $X_t=h_t+g_t$.   
\begin{itemize}
\item[1.] Obtain an initial component $\hat{h}_t^{(0)}=X_t-\mathcal{G}_t^\tau\big(X_t\big)$.  
	
\item[2.] Iterate, until convergence, the following step for $k=0,1,\ldots,$ 
\[
\hat{h}_t^{(k+1)}=\hat{h}_t^{(k)}-\mathcal{G}_t^\tau\big(\hat{h}_t^{(k)}\big).
\]
\item[3.] Take the converged estimate as the extracted component for $h_t$.  
\end{itemize}
We have some remarks regarding the aforementioned algorithm. (a) Choice of $\mathcal{G}_t^\tau$: It is feasible to use various choices of $\mathcal{G}_t^\tau$ including some central measures introduced in Section 2.2, which is the main benefit of utilizing ensemble patch transformation. To be specific, $\mbox{EAve}_t^\tau$ or $\mbox{EM}_t^\tau$ can be used for $\mathcal{G}_t^\tau$. (b)   Choice of $\tau$: The size parameter $\tau$ corresponds to a period in time domain. Thus, the parameter $\tau$ plays a crucial role in quality of the extracted low-frequency component. Selection method of $\tau$ will be discussed later.  

We now discuss a convergence property of the above algorithm under some conditions.  
\begin{thm}\label{thm1}
Suppose that we observe a real-valued sequence $(X_t)_t$ from a model $X_t=h_t+g_t$, where $\{h_t\}$, $t \in \mathbb{R}$ is a periodic sequence with $h_{t}=h_{t+\tau_0}$ and $\int_0^{\tau_0}h_t=0$, and $g_t$ is a signal such that $|G(\omega)|=0 , ~\omega \in \big\{\omega:\omega=\frac{2\pi k}{\tau_0}\pm 2n\pi,~\textup{for all } k=1,\dots,\tau_0-1~\textup{and } n\in \mathbb{N}\big\}$ and  $G(\omega)$ denotes Fourier transform of $g_t$. Then, for a given $\tau_0$, we obtain that $\hat{h}_t^{(k)} \to h_t ~\textup{as } k \to \infty$, where $\hat{h}_t^{(k+1)}=\hat{h}_t^{(k)}-\textup{EAve}_t^{\tau_0}\big(\hat{h}_t^{(k)}\big)$,  $\hat{h}_t^{(0)}=X_t-\textup{EAve}_t^{\tau_0}(X_t)$.
\end{thm}
\begin{proof}
$\textup{EAve}_t^{\tau_0}(X_t)$ can be expressed as 
\[
\textup{EAve}_t^{\tau_0}(X_t)=\phi^{\tau_0}_t * \phi^{\tau_0}_t * X_t, 
\]
where $\phi^{\tau_0}_t$ is a rectangular (boxcar) function defined as
\[
\phi^{\tau_0}_t=\left\{
\begin{array}{ll}
\frac{1}{\tau_0}, &~|t|<\tau_0\\
0, &~\mbox{otherwise}.  		
\end{array}
\right.
\]
Let $\xi^{\tau_0}_t=\phi^{\tau_0}_t * \phi^{\tau_0}_t$. Then $\hat{h}_t^{(k)}$ can be expressed as 
\[
\hat{h}_t^{(k)}=(\delta_t-\xi^{\tau_0}_t)^{*k}*X_t, 
\]
where $\delta_t$ denotes Kronecker delta function and $u^{*k}={\underbrace{u*u*\dots*u}}$ denotes convolution power. In addition,  $\Xi^{\tau_0}(\omega)={\cal F}\{\xi^{\tau_0}_t\}$ can be expressed as 
\begin{align*} 
\Xi^{\tau_0}(\omega)=\begin{cases}
1, & \omega=0,\pm 2\pi, \pm 4\pi, \dots \\ 
\left(\frac{\sin(\frac{\tau_0\omega}{2})}{\tau_0\sin(\frac{\omega}{2})}\right)^2,  & \omega \neq 0,\pm 2\pi, \pm 4\pi, \dots. 
\end{cases}
\end{align*}
Thus, it follows that $0<1-\Xi^{\tau_0}(\omega)<1$ for $\omega \notin \Big\{0\pm 2n\pi,\frac{2\pi}{\tau_0}\pm 2n\pi,\dots,\frac{2(\tau_0-1)\pi}{\tau_0}\pm 2n\pi\Big\}$. Furthermore, from the assumption of $|G(\omega)|=0$ for$~\omega \in \Big\{0\pm 2n\pi,\frac{2\pi}{\tau_0}\pm 2n\pi,\dots,\frac{2(\tau_0-1)\pi}{\tau_0}\pm 2n\pi\Big\}$, we conclude that
\[
\big|\big(1-\Xi^{\tau_0}(\omega)\big)^k G(\omega)\big| \to 0 
\]
as $k \to \infty$. 
\end{proof}

We close this section with a comparison of the proposed decomposition procedure with EMD for a better understanding of our procedure. Figure~\ref{ptemd} describes the difference of both procedures, which focuses on the local behavior of a signal
$
X_t = \cos(90 \pi t) + \cos(10 \pi t)
$
 in a particular time domain $t \in [0.5, 0.7]$. The essential step of EMD procedure is identifying the local extrema, and obtaining the upper and lower envelopes by interpolating the local maxima and minima, as shown in the right first panel of Figure~\ref{ptemd}. The corresponding mean envelope represents the local low-frequency mode effectively, and a signal is separated as residue and high-frequency mode by repeatedly removing lower frequency mode. From the results in the right second and third panels of Figure~\ref{ptemd}, it is necessary that the local extrema represent the local behavior of high-frequency component of a signal properly. In the case that the local behavior of high-frequency component is not distinct in a signal, EMD fails to decompose a signal like the case in Figure~\ref{signal3decom}. On the other hand, the proposed procedure takes a different approach of suppressing high-frequency component. It is not required to identify some local structure of the high-frequency component. Instead, the local high-frequency pattern is seized by the ensemble  upper and lower envelopes, as shown in the left first panel of Figure~\ref{ptemd}; thus, the ensemble mean suppresses the oscillating local pattern, and at the same time, represents the lower-frequency pattern in the left second panel of Figure~\ref{ptemd}. 

\begin{figure}[!t]
\centerfig{0.75 \columnwidth}{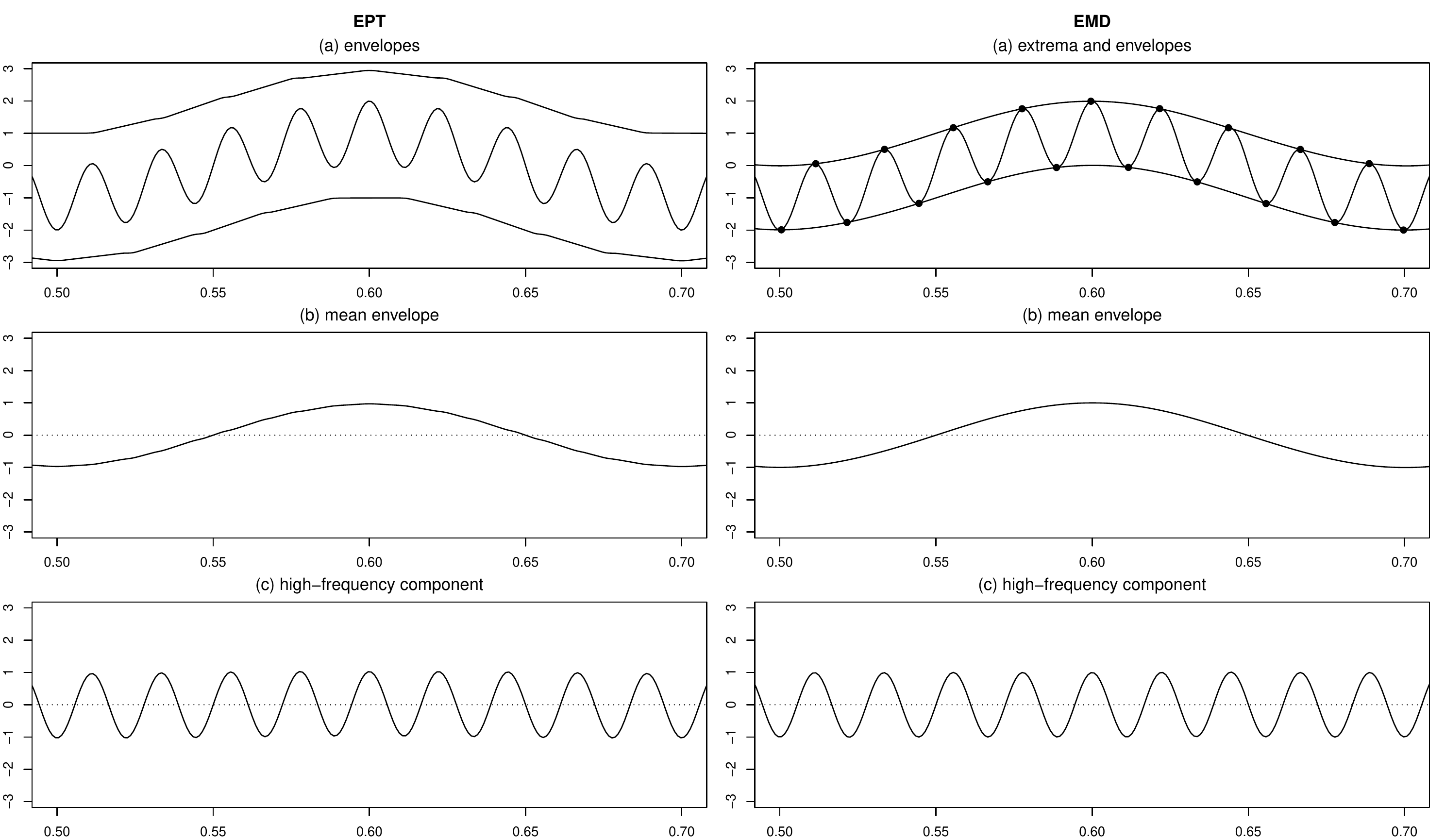}
\caption{Comparison of decomposition procedure by ensemble patch filtering and EMD.}
\label{ptemd}
\end{figure}

\subsection{Decomposition Procedure for Discrete Signals}
In this section, we focus on equally spaced discrete time signals and discuss some properties of the proposed decomposition method.   
\begin{dfn}
Let $\{X_i\}_{i=-\infty}^{\infty}$ be an equally spaced sequence such that $X_i:=X(t_i)$, where $X(t)$ denotes a continuous-time signal, $t_i=iT$ is sampling instant, and $T$ denotes sampling interval. For any $i$, the iterative representation of $X_i$ with filter ${\cal M}$ is defined as $IR_{i,\infty}(\{X_i\},{\cal M}):=\lim_{k\rightarrow \infty} IR_{i,k}(\{X_i\},{\cal M})$, where $IR_{i,k}(\{X_i\},{\cal M})=IR_{i,k-1}(\{X_i\},{\cal M})+{\cal M}(X_i-IR_{i,k-1}(\{X_i\},{\cal M}))$ and $IR_{i,1}(\{X_i\},{\cal M})={\cal M}X_i$. Furthermore, the sequence $\{X_i\}$ is said to be (iteratively) representable with filter ${\cal M}$ if $IR_{i,\infty}(\{X_i\},{\cal M})=X_i$ for all $i$. 
\end{dfn}

We note that if $X_i={\cal M}X_i$ for all $i$, then $\{X_i\}$ is (iteratively) representable with filter $\cal M$. 
\begin{dfn}
Let $\{X_i\}_{i=-\infty}^{\infty}$ be an equally spaced sequence such that $X_i:=X(t_i)$, where $X(t)$ denotes a continuous-time signal, $t_i=iT$ is sampling instant, and $T$ denotes sampling interval. Suppose that $\{X_i\}$ consists of two components as $X_i=h_i+g_i$ for all $i$. The component $\{h_i\}$ is said to be cancellable from $\{X_i\}$ with filter ${\cal M}$ if ${\cal M} X_i= {\cal M}g_i$ for all $i$. 
\end{dfn}

We then have the following result. 
\begin{thm}
Let $\{X_i\}_{i=-\infty}^{\infty}$ be an equally spaced sequence such that $X_i:=X(t_i)$, where $X(t)$ denotes a continuous-time signal, $t_i=iT$ is sampling instant, and $T$ denotes sampling interval. Suppose that $\{X_i\}$ consists of two components as $X_i=h_i+g_i$ for all $i$. Assume that (i) $\{g_i\}$ is (iteratively) representable with filter $\cal M$, and (ii) $\{h_i\}$ is cancellable from $\{X_i\}$ with filter $\cal M$. Then, it follows that $X_i-IR_{i,\infty}(\{X_i\},{\cal M})=h_i$ for all $i$. 
\end{thm}

A proof is directly obtained by both definitions. The following result can be considered as an analogy of Theorem~\ref{thm1} with equally spaced discrete time signals. 
\begin{lem}
\label{lemma}
Let $\{X_i\}_{i=-\infty}^{\infty}$ be an equally spaced sequence such that $X_i:=X(t_i)$, where $X(t)$ denotes a continuous-time signal, $t_i=iT$ is sampling instant, and $T$ denotes sampling interval. Suppose that $|X(\omega)|=0$ for $\omega \in \big\{\omega:~\omega=\frac{2\pi k}{\tau_0}\pm 2n\pi,\mbox{ for all}~k=1,\ldots,\tau_0-1 \mbox{ and } n\in\mathbb{N}\big\}$, where $X(\omega)$ is the Fourier transformation of $X_i$. Define the filter $\cal M$ as 
$
{\cal M} X_i = \textup{EAve}_i^{\tau_0}(X_i).
$
Then, the sequence $\{X_i\}$ is (iteratively) representable with filter $\cal M$. 
\end{lem}
A proof of Lemma 3.3 is easily obtained from proof of Theorem 3.1; hence, we omit it.

We remark that suppose that $\{X_i\}$ consists of two components as $X_i=h_i+g_i$ for all $i$. If $\cal M$ is a linear filter with ${\cal M} h_i =0$ for all $i$, then the sequence $\{h_i\}$ is cancellable from $\{X_i\}$ with filter $\cal M$. Then, we obtain the following result that extends the convergence property of Lemma~\ref{lemma} with ${\cal M} X_i = \textup{EAve}_i^{\tau_0}(X_i)$ to a general linear filter ${\cal M}$ under some conditions.  
\begin{cor}
Let $\{X_i\}_{i=-\infty}^{\infty}$ be an equally spaced sequence such that $X_i:=X(t_i)$, where $X(t)$ denotes a continuous-time signal, $t_i=iT$ is sampling instant, and $T$ denotes sampling interval. Suppose that $\{X_i\}$ consists of two components as $X_i=h_i+g_i$ for all $i$. Define the filter $\cal M$ as 
$
{\cal M} X_i = \textup{EAve}_i^{\tau_0}(X_i).
$
Assume that 
\begin{enumerate}
	\item[(i)] $|G(\omega)|=0$ for $\omega \in \big\{\omega:~\omega=\frac{2\pi k}{\tau_0}\pm 2n\pi,\mbox{ for all}~k=1,\ldots,\tau_0-1 \mbox{ and } n\in\mathbb{N}\big\}$, where $G(\omega)$ is the Fourier transformation of $g_i$.
	\item[(ii)] $\{h_i\}$ satisfies $h_i=h_{i+\tau_0}$ and $\sum_{i=1}^{\tau_0}h_i=0$. 
\end{enumerate}
Then, we obtain that  $X_i-IR_{i,\infty}(\{X_i\},{\cal M})=h_i$ for all $i$. 
\end{cor}

As for a final remark, we consider a simple example with designing an ideal filter that provides a strength of our method.  Suppose that we have a signal $X_i=h_i+g_i$, where $\{h_i\}$ satisfies $h_i=h_{i+3}$ and $\sum_{i=1}^{3}h_i=0$, and $\{g_i\}$ is a signal whose value suddenly changes from $-1$ to $1$ at $i=0$ as in Table \ref{tb:gi}. 
\begin{table}[t]
	\centering
	\caption{Siginal $\{g_i\}$}\label{tb:gi}
	\begin{tabular}{|c|c|c|c|c|c|c|c|c|c|}
		\hline
		\textbf{$i$} & $\dots$ & -3 & -2 & -1 & 0 & 1 & 2 & 3 & $\dots$ \\ \hline
		$g_i$       & $\dots$          & -1          & -1 & -1 & 1 & 1 & 1 & 1 & $\dots$ \\ \hline
	\end{tabular}
\end{table} 
We define the filter ${\cal M}$ as 
$
{\cal M} X_i = \textup{EAve}_i^{\tau_0}(X_i).
$
Note that $\{h_i\}$ is cancellable from $\{X_i\}$ with filter $\cal M$, but $\{g_i\}$ cannot be (iteratively) representable with filter ${\cal M}$ since $|G(\omega)|>0$ for some  $\omega \in \big\{\omega:~\omega=\frac{2\pi k}{3}\pm 2n\pi,\mbox{ for all }~k=1,2 \mbox{ and } n\in\mathbb{N}\big\}$, where $G(\omega)$ is the Fourier transformation of $g_i$. Thus, it is not able to obtain $h_i$ from $X_i-IR_{i,\infty}(\{X_i\},{\cal M})$. This is because the filters of the moving average class are not suitable for expressing data with a sharp mean change such as $\{g_i\}$. It is generally known that data with such a sharp mean change can be easily represented by a median filter (Gallagher and Wise, 1981). In particular, $\{g_i\}$ used in the example is a root signal since it does not change even if it passes through the median filter repeatedly. Hence, the convergence property is ensured. In summary, $ \mbox{Ave}^{\tau}_t (X_t)$ is advantageous to cancel $\{h_i\} $, but it cannot represent $\{g_i\}$ properly. On the other hand, $\mbox{Med}^{\tau}_t (X_t)$ is not capable of canceling $\{h_i\} $, but is useful for expressing $\{g_i\}$.  As a result, a combination of both filters might lead to desired decomposition results, which is feasible  under the ensemble patch transform framework, not just patch transform one. It is a benefit of the proposed transformation. We now consider a filter
$
{\cal M}^* X_i = \mbox{median}\left(\mbox{Ave}_{i+\ell}^{\tau_0}(X_i)\right).
$
Due to the property of the linear filter and the condition $\sum_{i = 1}^{3} h_i = 0$, it follows that $\mbox{Ave}_{i+\ell}^{\tau_0}(X_i)=\mbox{Ave}_{i+\ell}^{\tau_0}(h_i)+\mbox{Ave}_{i+\ell}^{\tau_0}(g_i)=\mbox{Ave}_{i+\ell}^{\tau_0}(g_i)$. So, this filter separates $h_i$ and $g_i$, and cancels $h_i$. In the example, the value of $ \mbox{Ave}_{i + \ell}^{\tau_0}(g_i)$ for each $\ell$ is listed in Table \ref {tb:calres}. Then, by passing the median filter as the second filter, we obtain a signal $\{\ \dots, -1, -2 / 3,2 / 3,1, \dots \}$, which completely represents $h_i$ except $ i \in \{- 1,0 \}$. An iterative calculation of $IR_{i,k}(\{X_i\},{\cal M}^*)$ $k=2,3\ldots$ using filter ${\cal M}^*$ provides the result in Table \ref{tb:calres}.
\begin{table}[!t]
	\centering
	\caption{Results for  $\mbox{Ave}_{i+\ell}^{\tau_0}(g_i)$, $\mbox{median}\left(\mbox{Ave}_{i+\ell}^{\tau_0}(X_i)\right)$, $IR_{i,k}(\{X_i\},{\cal M})$, $k=2,3,\ldots$.}
	\begin{tabular}{|c|c|c|c|c|c|c|c|c|c|c|}
		\hline
		\multicolumn{2}{|c|}{$i$} & $\dots$ & -3 & -2 & -1 & 0 & 1 & 2 & 3 & $\dots$ \\ \hline
		\multirow{3}{*}{$\mbox{Ave}_{i+\ell}^{\tau_0}(g_i)$} & $\ell=-1$ & $\dots$ & -1 & -1  & -1 & -2/3 & 2/3 & 1 & 1 & $\dots$ \\ \cline{2-11} 
		& $\ell=0$  & $\dots$ & -1 & -1   & -2/3 & 2/3  & 1 & 1 & 1 & $\dots$ \\ \cline{2-11} 
		& $\ell=1$  & $\dots$ & -1 & -2/3 & 2/3  & 1    & 1 & 1 & 1 & $\dots$ \\ \hline
		\multicolumn{2}{|c|}{$\mbox{median}\left(\mbox{Ave}_{i+\ell}^{\tau_0}(X_i)\right)$} & $\dots$ & -1 & -1 & -2/3   & 2/3 & 1 & 1 & 1 & $\dots$ \\ \hline
		\multicolumn{2}{|c|}{$IR_{i,2}(\{X_i\},{\cal M})$} & $\dots$ & $h_i$ & $h_i$ & $h_i-1/3$ & $h_i+1/3$ & $h_i$ & $h_i$ & $h_i$ & $\dots$ \\ \hline
		\multicolumn{2}{|c|}{$IR_{i,3}(\{X_i\},{\cal M})$} & $\dots$ & $h_i$ & $h_i$ & $h_i-1/3$ & $h_i+1/3$ & $h_i$ & $h_i$ & $h_i$ & $\dots$ \\ \hline
		\multicolumn{11}{|c|}{$\vdots$} \\ \hline
		\multicolumn{2}{|c|}{$IR_{i,\infty}(\{X_i\},{\cal M})$} & $\dots$ & $h_i$ & $h_i$ & $h_i-1/3$ & $h_i+1/3$ & $h_i$ & $h_i$ & $h_i$ & $\dots$ \\ \hline
	\end{tabular}
	\label{tb:calres}
\end{table}
We remark that in the example, the difference between $\cal M$ and ${\cal M}^*$ is found in the index set where $\{g_i\}$ can be perfectly represented. As the iteration progresses, the index set where $\{g_i\}$ is fully expressed by $X_i-IR_{i, k}({X_i}, {\cal M})$ converges to $\emptyset$ as $k \rightarrow \infty$, while the index set that $\{g_i\}$ is perfectly represented by $X_i-IR_ {i, k} ({X_i}, {\cal M}^*) $ converges to $\mathbb{Z} \setminus\{- 1 , 0 \}$.

\section{Numerical Study} \label{sec.simulation}

Here we conduct a numerical study and discuss its results to assess the practical performance of the proposed method. In this numerical study, we compare the proposed method with EMD, and discuss the merits of the proposed method over EMD. The proposed method is implemented by the algorithm introduced in Section~3.2. Various type of a generic central measure $\mathcal{G}_t^\tau$ can be applied for patch transform. Ensemble average $\mbox{EAve}_t^\tau$ is used for Examples~1 and 3, ensemble median of patch transform by average, median(Ave$^\tau_{t+\ell}(X_t)$) is for Example 2, and ensemble mean envelope $\mbox{EM}_t^\tau$ is for Example 4.


\subsection{Example 1: Composite Sinusoidal Signal}
Suppose that we have 1000 equally spaced observations from a synthetic test signal $X_t= \cos(90 \pi t) + \cos(10 \pi t)$, $t\in [0,1]$ in Figure~\ref{signal1}. 
It is expected that both EMD and the proposed decomposition by ensemble patch transform (EPT) work well. The difference of frequencies of two components are large enough to identify the local pattern of high-frequency component by EMD. By taking a suitable size parameter $\tau=21$, the proposed method separates two components efficiently as well. Figure~\ref{signal1decom} shows the decomposition results by EMD, EEMD and the proposed method, which imply that all methods work properly to decompose the signal. 


\begin{figure}[!t]
\centerfig{0.95\columnwidth}{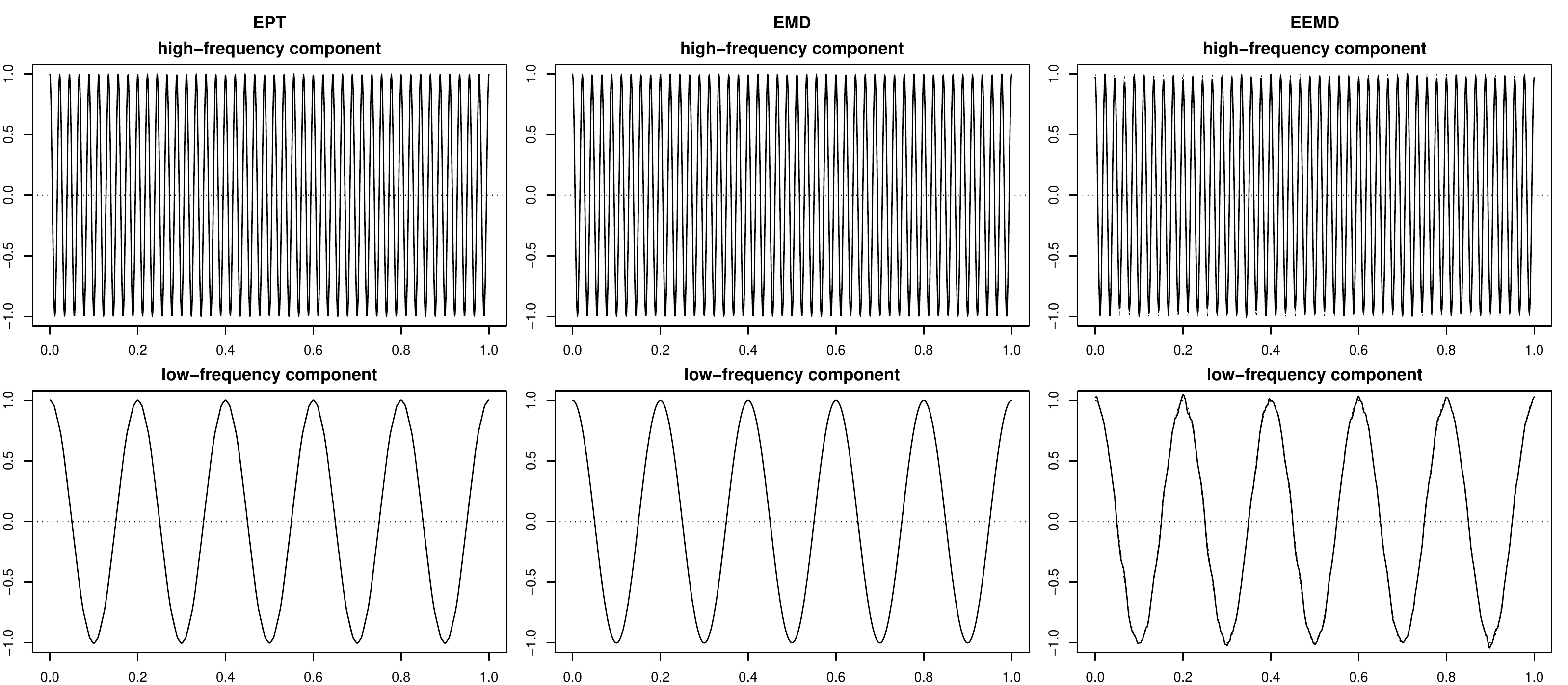}
\caption{Decomposition of test signal $X_t=\cos(90 \pi t) + \cos(10 \pi t)$. From the left to right, the decomposition results by the proposed method, EMD and EEMD, respectively.}
\label{signal1decom}
\end{figure}

\subsection{Example 2: Piecewise Signal}
We consider a non-stationary piecewise signal that consists of a low-frequency component and a high-frequency component piecewisely defined as $X_t = \cos(90 \pi t)I(t \le 0.5) + \cos(10 \pi t)I(t > 0.5), \ t \in [0,1]$ shown in Figure~\ref{signal2}. Huang et al. (1998) and Huang et al. (2003) pointed out that EMD fails to decompose a signal with mode mixing, which means that different modes of oscillations coexist in a single intrinsic mode function (IMF). On the other hand, the proposed method is able to locally suppress the high-frequency mode whose period is less than some size parameter. The dotted line and solid line of Figure~\ref{signal2decom} represent true components and extracted components by each method. From the results, we observe that the proposed method performs better than EMD and EEMD. Here we use size parameter $\tau=21$ for our method. 
\begin{figure}[!t]
\centerfig{0.85 \columnwidth}{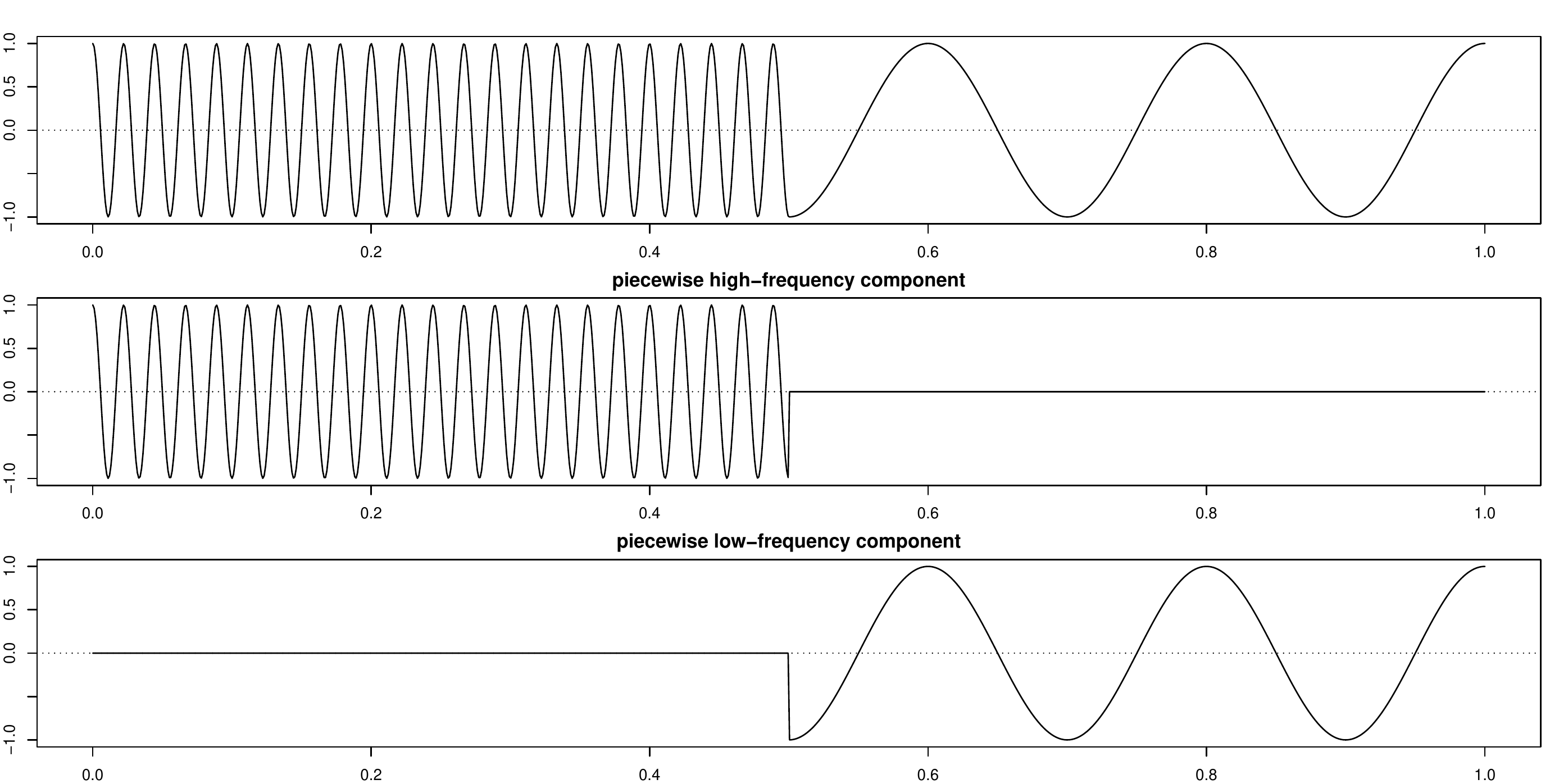}
\caption{Test signal $X_t=\cos(90 \pi t)I(t \le 0.5) + \cos(10 \pi t)I(t > 0.5)$ and its two piecewise components.}
\label{signal2}
\end{figure}

\begin{figure}[!t]
\centerfig{0.95 \columnwidth}{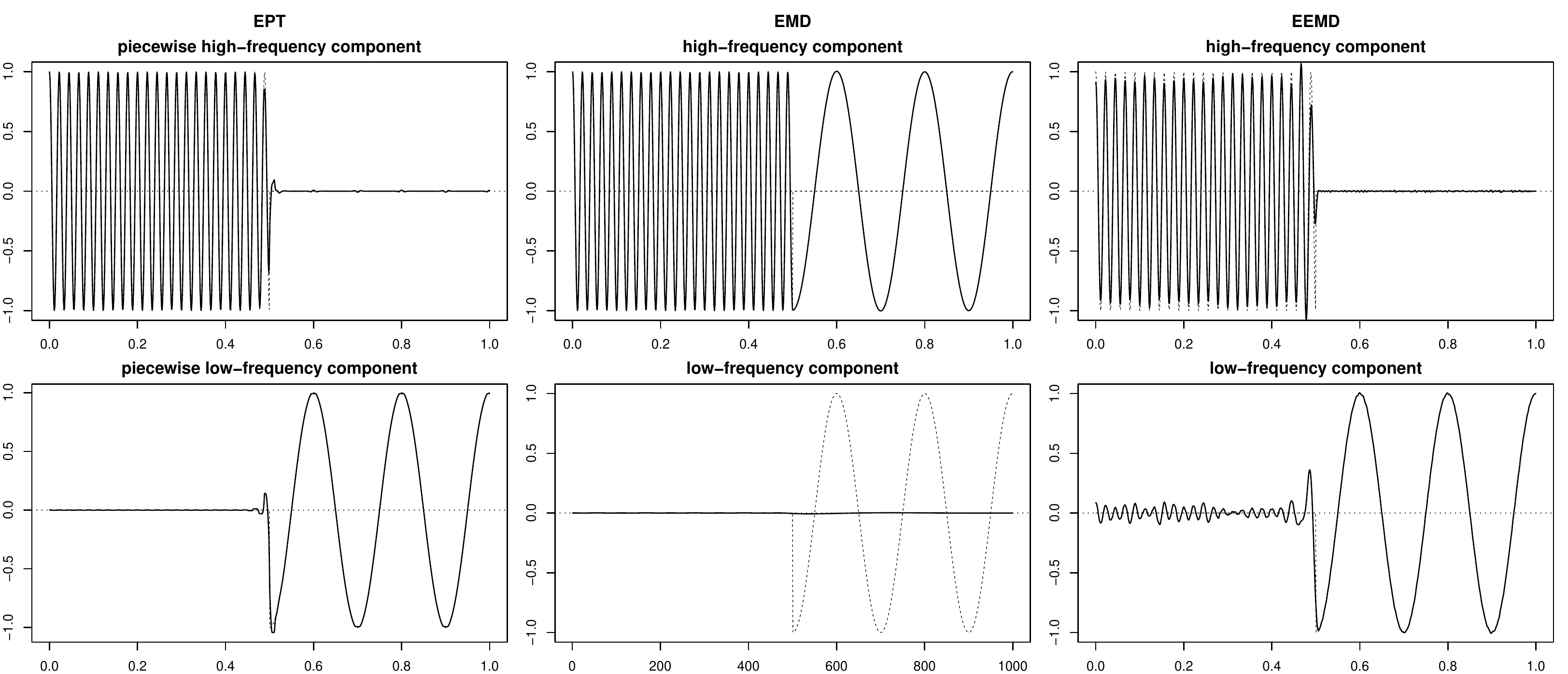}
\caption{Decomposition of test signal $X_t=\cos(90 \pi t)I(t \le 0.5) + \cos(10 \pi t)I(t > 0.5)$. From the left to right, the decomposition results by the proposed method, EMD and EEMD, respectively.}
\label{signal2decom}
\end{figure}




\subsection{Example 3: Noisy Signal}
We evaluate a robustness of the proposed decomposition to noise signals. We generate  a noisy signal $X_t+\epsilon_t$,  where  $X_t=\cos(90 \pi t) + \cos(10 \pi t)$ is the signal in Figure~\ref{signal1} and $\epsilon_t$ denote Gaussian errors with signal-to-noise ratio 7. The decomposition results by the proposed method, EMD and EEMD are shown in Figure~\ref{noisedecom}. As one can see, EMD is sensitive to noises. In fact, the effect of non-informative fluctuation distorts the subsequent decomposition results of EMD, which is due to interpolation process in the construction of envelopes based on local extrema. On the other hand, the proposed method is robust to the noises since the  decomposition is precessed without the identification of fluctuations. The decomposition results of  Figure~\ref{noisedecom} support this fact. If we regard noise as fluctuation with the highest frequency, the proposed method with relatively small $\tau$ might separate a noise from a signal. By taking the size parameter $\tau=10$, a noisy signal is decomposed as the highest component of noise and the low-frequency residue component, which corresponds to a signal $X_t$. This low-frequency residue component is repeatedly decomposed with the size parameter $\tau=21$.  We notice that EEMD performs well for decomposition. 
\comment{
\begin{figure}[!t]
\centerfig{0.85\columnwidth}{ex1noise.pdf}
\caption{Noisy test signal $X_t^3$.}
\label{noise}
\end{figure}
}
\begin{figure}[!t]
\centerfig{0.95\columnwidth}{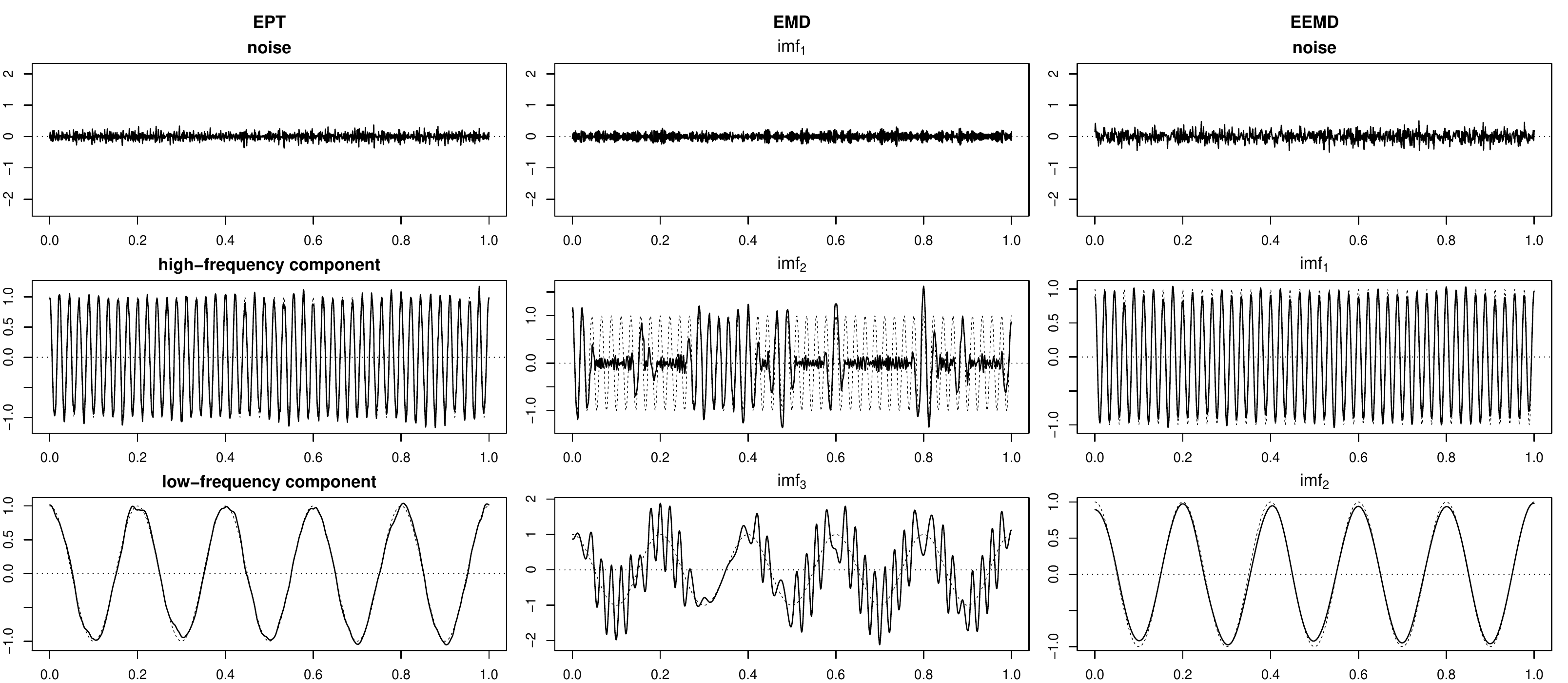}
\caption{Decomposition of noisy signal $X_t+\epsilon_t, X_t=\cos(90 \pi t) + \cos(10 \pi t)$. From the left to right, the decomposition results by the proposed method, EMD and EEMD, respectively.}
\label{noisedecom}
\end{figure}

\begin{figure}[!t]
\centerfig{0.95\columnwidth}{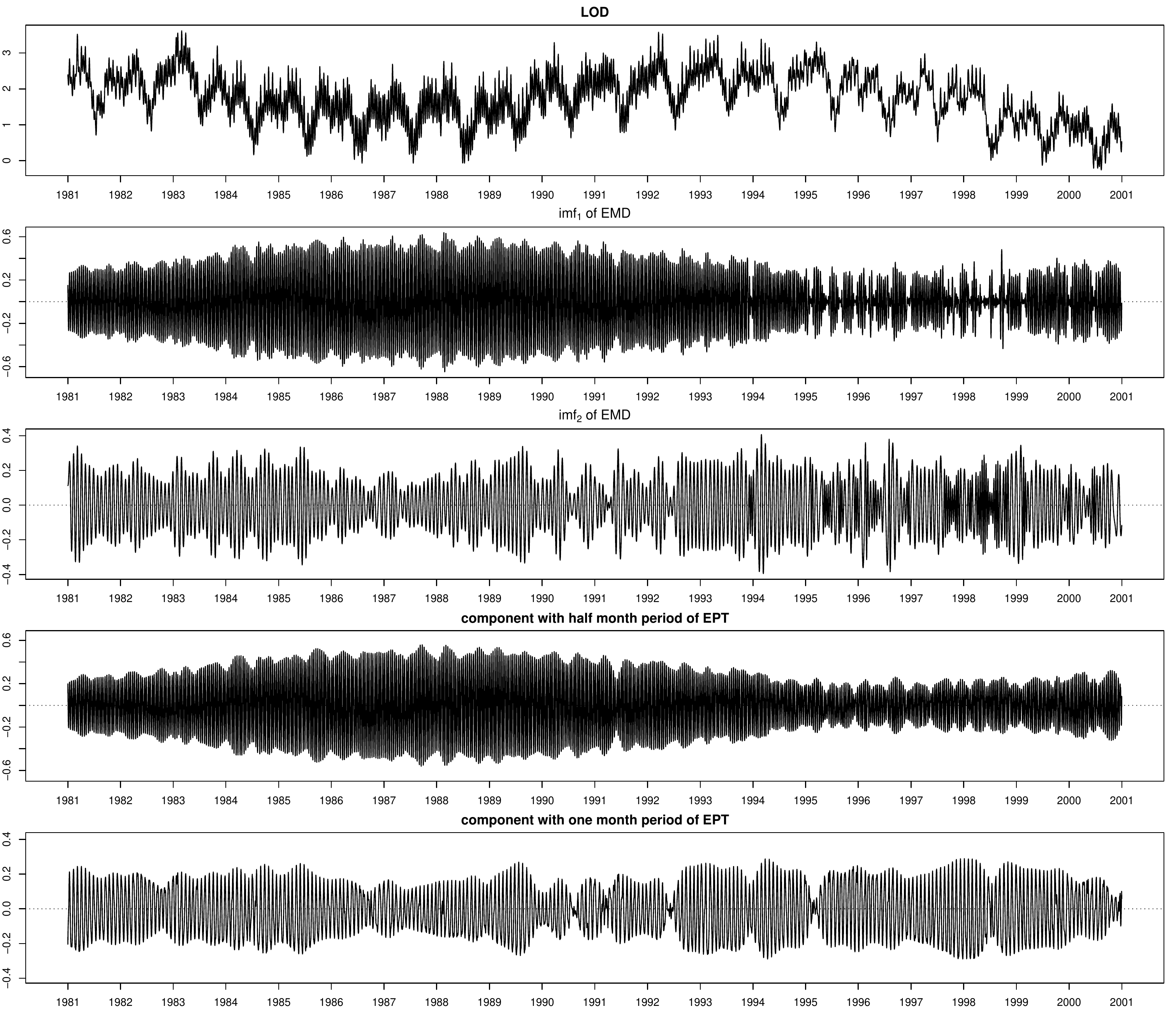}
\caption{LOD signal.}
\label{LOD}
\end{figure}

\subsection{Example 4: Length-of-Day Data}
For further demonstrating the usefulness of the proposed method, we decompose the length-of-day (LOD) data analyzed in 
Huang et al. (2003). The dataset {\tt{comb200\_daily.eop}} is available in the ftp site {\tt{ftp://euler.jpl.nasa.gov/keof/combinations/2000}}. For analysis, we use 7,305 samples that cover the period of 1 January 1981 to 31 December 2000, which are shown in the first row of Figure~\ref{LOD}. LOD is measured in millisecond. The second and third rows show two IMFs obtained by EMD, which represent certain oscillation modes of half month period and one month period. However, we observe mode mixing between year 1995 and year 1999. The proposed method separates these mixing modes successfully as shown in the fourth and fifth rows of Figure~\ref{LOD}. Each of the components has a clear physical meaning, which reveals the fluctuation mechanism of the LOD signal.


\section{Selection of Size Parameter} \label{sec.selection}

Here we discuss the selection method of size parameter $\tau$ for ensemble patch transformation. We propose two selection methods of the size parameter $\tau$. One is performed in a priori way, and the other is based on the posteriori information of the decomposition.

For the first method, we point out that the size parameter $\tau$ corresponds to a period in time domain. When a priori information of periodic pattern of a signal is available, a selection of the size parameter can be conducted based on the distribution of periodic pattern. Such information can be obtained through the empirical periods of a distance between local maxima (or local minima). Note that the empirical period is expressed by the number of observations between local maxima, not by the distance of physical time. Figure~\ref{signal1period} shows the distribution of empirical periods for a signal $X_t = \cos(90 \pi t) + \cos(10 \pi t)$ and its high-frequency component $\cos(90 \pi t)$, where the high frequency pattern is apparent in signal $X_t$. It seems that the dominated period is 21, which is set to be our estimated parameter, $\hat{\tau}=21$. In fact, the decomposition results in Section 4 are based on this selection method.   

\begin{figure}[!t]
\centerfig{0.75\columnwidth}{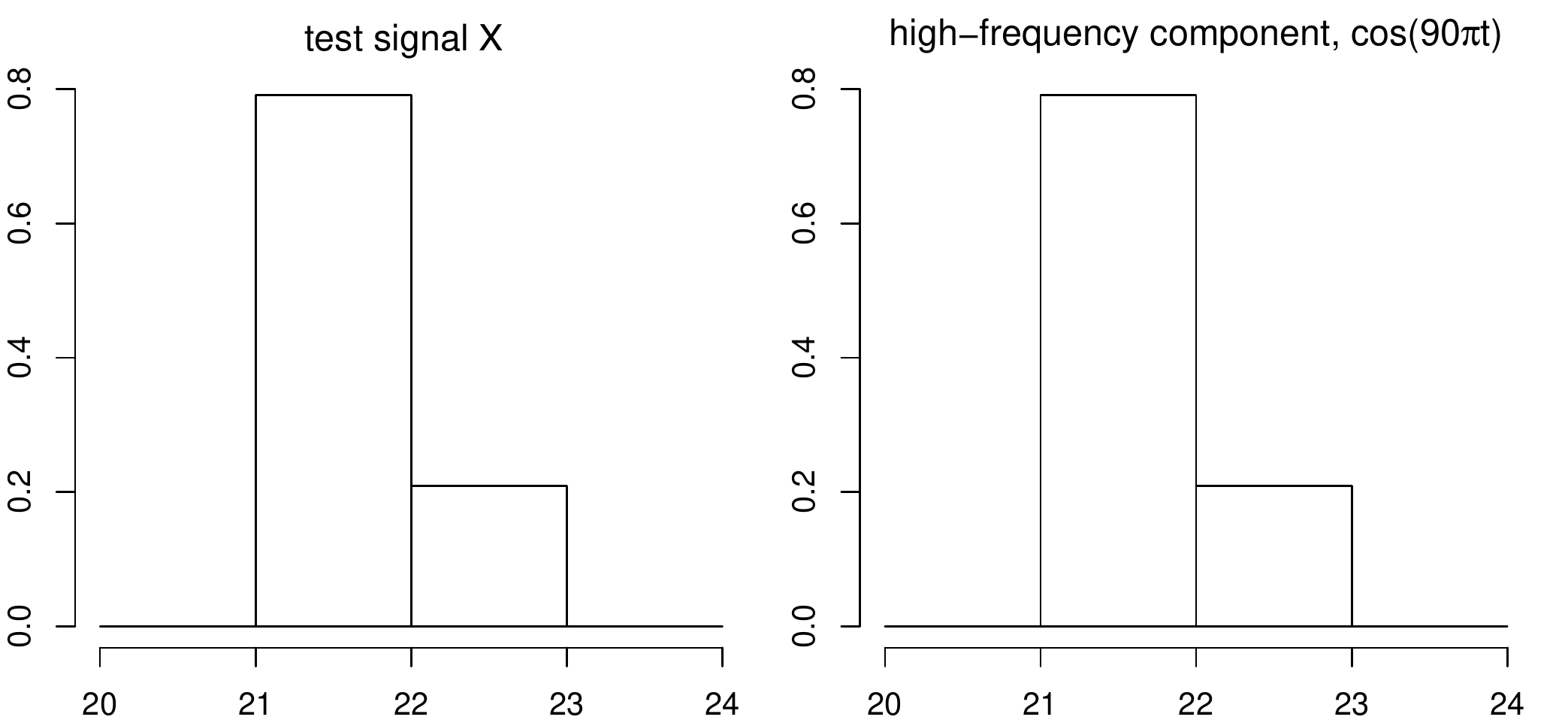}
\caption{The distribution of the empirical period for a signal $X_t = \cos(90 \pi t) + \cos(10 \pi t)$ and its component $\cos(90 \pi t)$.}
\label{signal1period}
\end{figure}

In the case that the frequency ratio of components composing a signal falls below a certain range, the local pattern of the high-frequency component may not be distinct; thus, the above selection method based on empirical periods is not appropriate. From the results in Figure 2, we observe that the proposed method might separate two components reasonably according to the frequencies. Nevertheless, the components should be weakly correlated to each other unless they are orthogonal. Hence, for the second method, we use correlation information between two components extracted by ensemble patch transformation. That is, through the grid search for a certain range of the size parameter, the size parameter $\tau$ is selected having the minimum correlation between the decomposed components. Figure~\ref{signal3corr} shows the sample correlations between the extracted components for the signal $X_t=\cos(100 \pi t) + 4 \cos(60 \pi t)$ in Figure 2 over a range of $\tau$, which produces $\hat{\tau}=16$. 

\begin{figure}[!t]
\centerfig{0.55\columnwidth}{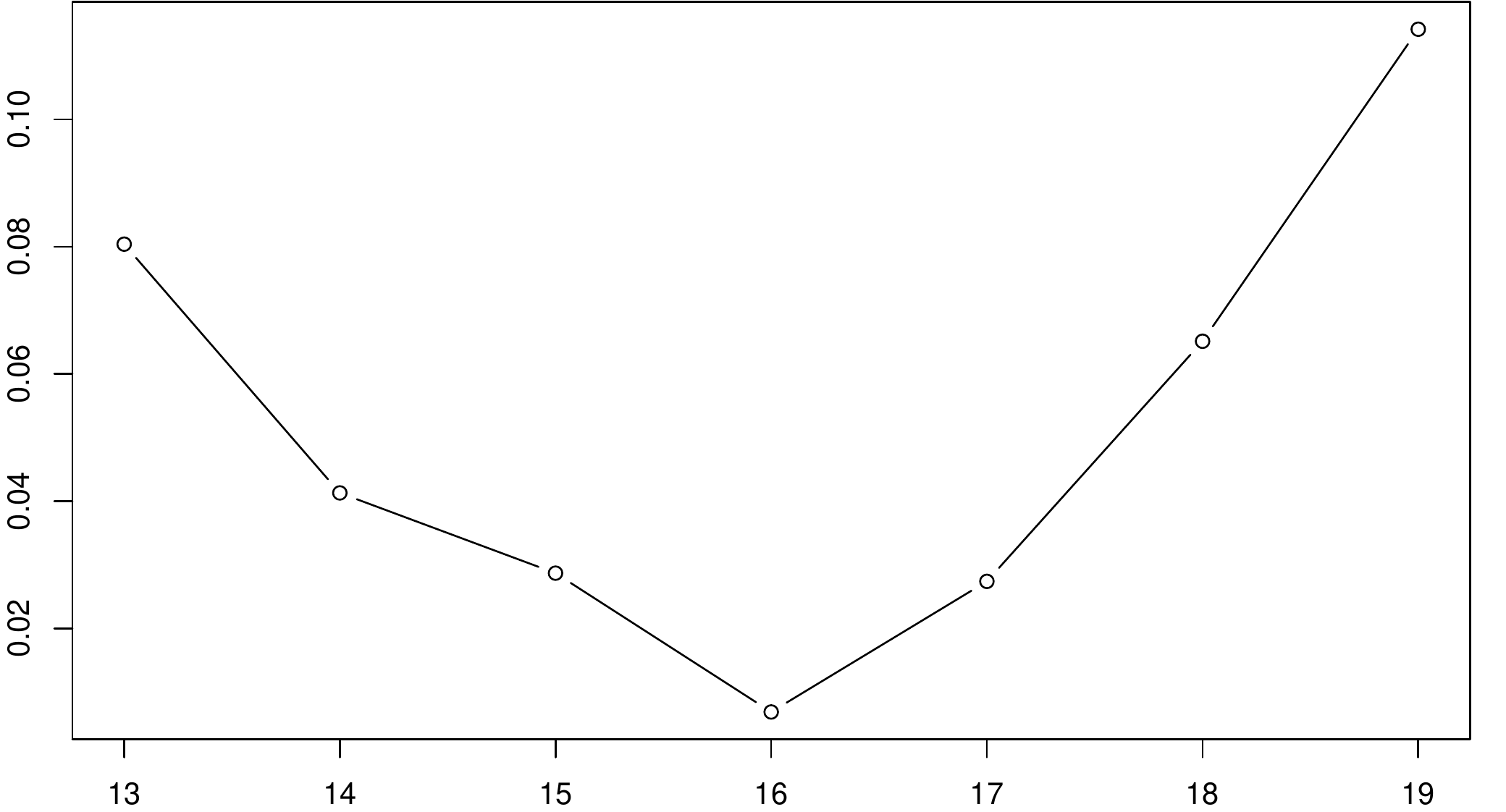}
\caption{Correlation between decomposed components of test signal $X_t=\cos(100 \pi t) + 4 \cos(60 \pi t)$ over a range of $\tau$.}
\label{signal3corr}
\end{figure}

We remark that through extensive experiments, we observe that our method is somewhat robust to the selection of size parameter. Suppose that we decompose the signal $X_t = \cos(90 \pi t) + \cos(10 \pi t)$ into two components by the proposed method with a range of  $\tau=18$ to 23.  Figure~\ref{signal1robust} shows the differences between the extracted high-frequency component and the true component $\cos(90 \pi t)$ over the range of $\tau$. As one can see, the results are robust to the choice of the size parameter $\tau$ value.

Finally, the proposed selection methods of the parameter $\tau$ lack a theoretical justification.  An objective way with theoretical backup might improve the performance and the practicality of the proposed method. This topic is left for future study.

\comment{
for selecting a proper size parameter, researcher may make use of physical information of which one want to interpret phenomenon.
As for LOD signal, fluctuations on a short scale of time can be extracted based on the known information of the periodic pattern with half month or one month. In fact, this information coincides with the empirical period of the LOD signal. The fourth and fifth row of Figure~\ref{LOD} is sequentially extracted with the size parameter $\tau=15$ and $30$.

In this case, the high frequency pattern is apparent in a signal $X^3_t$. The decomposition result of Figure~\ref{signal1decom} is produced by the size $\tau=21$. }
\begin{figure}[!t]
\centerfig{0.75\columnwidth}{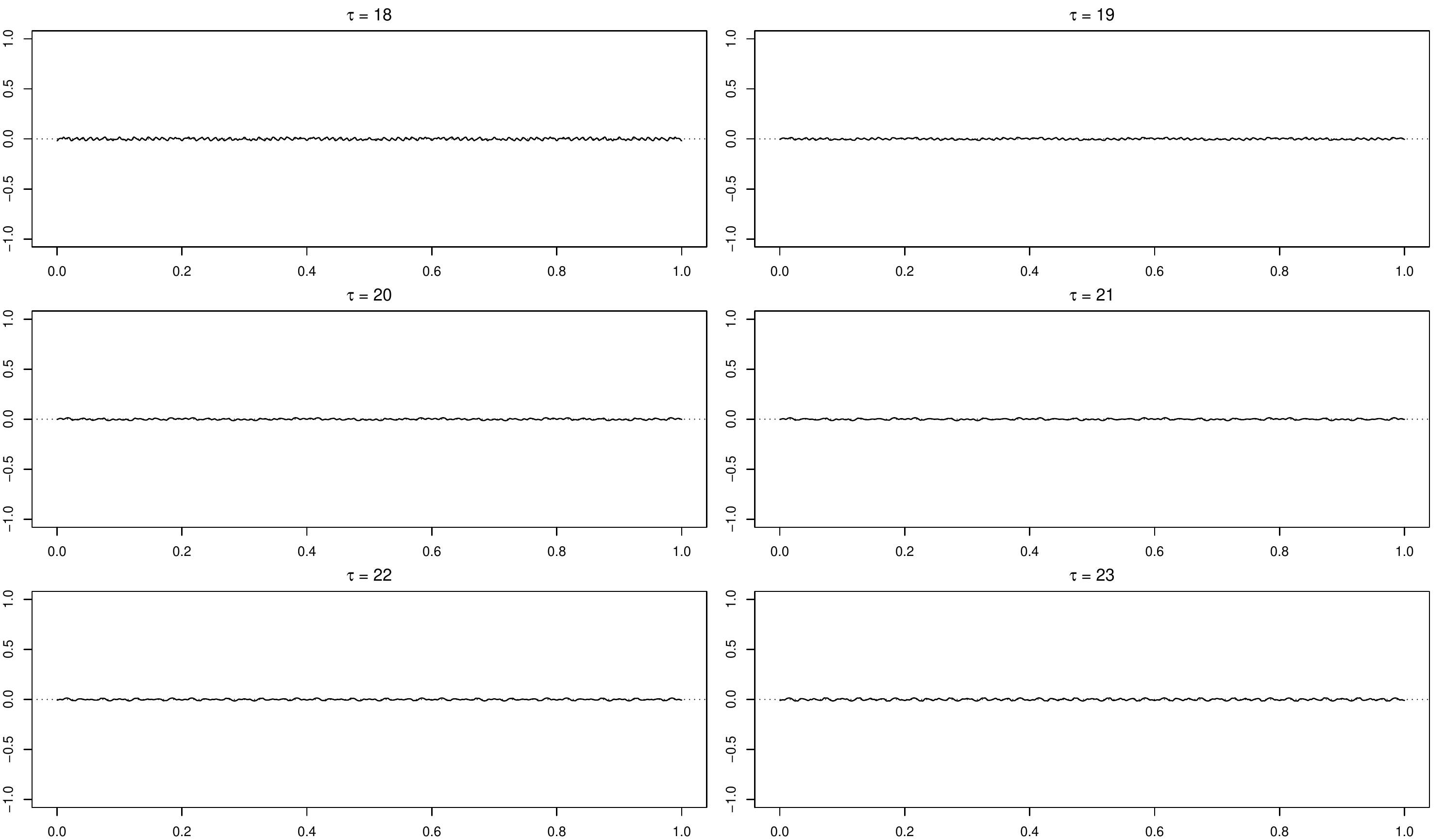}
\caption{The difference of true component $\cos(90 \pi t)$ and decomposed high-frequency component of test signal $X_t= \cos(90 \pi t) + \cos(10 \pi t)$ according to size parameter $\tau$.}
\label{signal1robust}
\end{figure}

\comment{
\subsection{Missing Values}
The multiscale decomposition methods implicitly assume that a signal is equally spaced. However, in practice, missing values occur quite common for many signals. For example, missing values may occur at random for intermittent wireless signal due to the malfunction of network device.
The proposed method is robust to missing values since its decomposition process identifies the lower-frequency component with suppressing the highly oscillating pattern and missing values have little effect on the highly oscillating pattern.

Consider a signal $X_t^3$ with 30\% missing values denoted by black points in the first row of Figure~\ref{signal1missdecom}.
The decomposition results by the proposed method in the left panel is comparable to the results by EMD in the right panel.

\begin{figure}[!ht]
\centerfig{0.95\columnwidth}{ex1mdecom.pdf}
\caption{Decomposition of test signal $X_t^3$ with 30\% missing values.}
\label{signal1missdecom}
\end{figure}
}

\section{Concluding Remarks} \label{sec.conclusion}

In this paper, we have introduced a new transformation technique, termed  `ensemble patch transformation' which is designed for  visualization and data analysis.  We have pointed out that this transformation can be used for filtering, and then proposed a decomposition method. The proposed decomposition procedure has taken a new viewpoint to handle the high-frequency component in that the high-frequency component is extracted by suppressing local oscillating pattern without the identification of high-frequency component, even though it shares a common principle with EMD where the high-frequency component is obtained by repeatedly removing lower frequency modes by sifting. We have presented an effective algorithm for implementation of the proposed method with some theoretical properties. The empirical performance of the proposed method has been evaluated throughout various numerical experiments and the analysis of real-world signal. Results from these experiments illustrate the proposed method possesses promising empirical properties.



\begin{thebibliography}{}


\bibitem[{Colominas(2014)}]{Colominas2014} Colominas, M.~A., Schlotthauer, G. and Torres, M.~E. (2014).  Improved complete ensemble EMD: A suitable tool for biomedical signal processing. \emph{Biomedical Signal Processing and Control}, {\bf 14}, 19--29.


\bibitem[{Cryer(2008)}]{Cryer2008} 
Cryer, J. D. and Chan, K.-S. (2008). {\it Time Series Analysis with Applications in R}, Second Edition, Springer,  New York. 


\bibitem[{Daubechies(2011)}]{Daubechies2011}  Daubechies, I.,  Lu, J. and Wu, H.~T. (2011). Synchrosqueezed wavelet transforms: An empirical mode decomposition-like tool. \emph{Applied and Computational Harmonic Analysis}, {\bf 30}, 243--261.

\bibitem[{Dragomiretskiy and Zosso(2014)}]{Dragomiretskiy2014}
Dragomiretskiy, K. and Zosso, D. (2014). Variational mode decomposition.  \emph{IEEE Transactions on Signal Processing}, {\bf 62}, 531--544. 

\bibitem[{Fryzlewicz(2011)}]{Fryzlewicz2011} Fryzlewicz, P. and Oh, H.-S. (2011). Thick pen transformation for time series.  \emph{Journal of the Royal Statistical Society} B, {\bf 73}, 499--529.

\bibitem[{Huang(1998)}]{Huang1998} 
Huang, N.~E., Shen, Z., Long, S.~R., Wu, M.~L., Shih, H.~H.,  Zheng, Q., Yen, N.~C., Tung, C.~C. and Liu, H.~H. (1998). The empirical mode decomposition and Hilbert spectrum for nonlinear and nonstationary time series analysis. \emph{Proceedings of the Royal Society London A.}, {\bf 454}, 903--995.

\bibitem[{Huang(2003)}]{Huang2003}
Huang, N.~E., Wu, M.~C., Long, S.~R., Shen, S., Qu, W., Gloerson, P. and Fan, K. L. (2003).  A confidence limit for the empirical mode decomposition and Hilbert spectral analysis. \emph{Proceedings of the Royal Society London A.}, {\bf 459}, 2317--2345.  

\bibitem[{Gallagher(1981)}]{Gallagher1981} Gallagher, N. and Wise, G. (1981). A theoretical analysis of the properties of median filters.  \emph{IEEE Transactions on Acoustics, Speech, and Signal Processing}, {\bf 29}, 1136--1141.

\bibitem[{Lindeberg(1994)}]{Lindeberg1994} Lindeberg, T. (1994). \emph{Scale-Space Theory in Computer Vision}, Kluwer, Boston.

\bibitem[{Meignen(2012)}]{Meignen2012}
Meignen, S., Oberlin, T. and McLaughlin, S. (2012). A new algorithm for multicomponent signals analysis based on synchrosqueezing: With an application to signal sampling and denoising. \emph{IEEE Transactions on Signal Processing}, {\bf 60}, 5787--5798.

\bibitem[{Rilling(2008)}]{Rilling2008} Rilling, G. and Flandrin, P. (2008).  One or two frequencies? The empirical mode decomposition answers. \emph{IEEE Transactions on Signal Processing}, {\bf 56}, 85--95.





\bibitem[{Thakur(2011)}]{Thakur2011}
Thakur, G. and  Wu, H.-T. (2011). Synchrosqueezing-based recovery of instantaneous frequency from nonuniform samples.  \emph{SIAM Journal on Mathematical Analysis}, {\bf 43}, 2078--2095. 

\bibitem[{Thakur(2013)}]{Thakur2013}
Thakur, G., Brevdo, E., Fuckar, N. S. and  Wu, H.-T. (2013). The synchrosqueezing algorithm for time-varying spectral analysis: Robustness properties and new paleoclimate applications. {\it Signal Processing}, {\bf 93}, 1079--1094. 

\bibitem[{Torres(2011)}]{Torres2011} Torres, M.~E.,  Colominas, M.~A.,  Schlotthauer, G. and Flandrin, P. (2011). A complete ensemble empirical mode decomposition with adaptive noise. \emph{Proceedings of 2011 IEEE International Conference on Acoustics, Speech and Signal Processing (ICASSP)}, 4144--4147.

\bibitem[{WuandHuang(2009)}]{WuandHuang2009} Wu, Z.  and Huang, N.~E. (2009). Ensemble empirical mode decomposition: a noise assisted data analysis method. \emph{Advances in Adaptive Data Analysis}, {\bf 1}, 1--41.

\bibitem[{Yeh(2010)}]{Yeh2010} Yeh, J.~R., Shieh, J.~S. and Huang, N.~E. (2010). Complementary ensemble empirical mode decomposition: A novel noise enhanced data analysis method. \emph{Advances in Adaptive Data Analysis}, {\bf 2}, 135--156.


















































\end{thebibliography}
\end{document}